\documentclass[11pt]{article}

\usepackage[margin=1in]{geometry}
\usepackage[cmex10]{amsmath}
\usepackage{amsthm} 
\usepackage{amssymb}
\usepackage{enumitem}
\usepackage{graphicx}
\usepackage{url}
\usepackage{refcount}
\usepackage{xcolor}

\newtheorem{proposition}{Proposition}
\newtheorem{lemma}{Lemma}

\newtheorem{theorem}{Theorem}

\newtheorem{claim}{Claim}

\newtheorem{axiom}{Axiom}

\newcommand{\wcapacity}{\textsc{Mwisl}}
\newcommand{\capswcapacity}{\textsc{MWISL}}
\newcommand{\mwisl}{\wcapacity}
\newcommand{\tlog}{\widehat{\log}}

\newcommand{\wscheduling}{\textsf{WScheduling}}
\newcommand{\uscheduling}{\textsf{UScheduling}}

\newcommand{\cG}{{\cal G}}
\newcommand{\cH}{{\cal H}}
\newcommand{\cE}{{\cal E}}
\newcommand{\cK}{{\cal K}}
\newcommand{\cI}{{\cal I}}
\newcommand{\cS}{{\cal S}}
\newcommand{\cP}{{\cal P}}
\newcommand{\cF}{{\cal F}}
\newcommand{\cC}{{\cal C}}
\newcommand{\s}{{\mathfrak l}}
\newcommand{\e}{{\mathfrak m}}
\newcommand{\Ds}{{\Delta}}
\newcommand{\ghi}{{G_{hi}}}
\newcommand{\glo}{{G_{lo}}}
\newcommand{\ghim}{{G_{hi}^M}}
\newcommand{\glom}{{G_{lo}^M}}

\newcommand{\mypar}[1]{{\medskip\textbf{#1.}}}

\begin{document}

\title{Effective Wireless Scheduling via Hypergraph Sketches
  \thanks{This work contains an extended treatment of results announced in \cite{us:stoc15}, \cite{us:fsttcs15}, and \cite{us:icalp17}.}}

\author{
  Magn\'us M. Halld\'orsson
  \qquad
  Tigran Tonoyan \\ \\
  ICE-TCS, School of Computer Science, Reykjavik University \\
  \url{{magnusmh,ttonoyan}@gmail.com}
}

\begin{titlepage}

\maketitle              

\begin{abstract}
An overarching issue in resource management of wireless networks is assessing their \emph{capacity}:
\emph{How much communication can be achieved in a network, utilizing all the tools available: power control, scheduling, routing, channel assignment and rate adjustment?}
We propose the first framework for approximation algorithms in the \emph{physical model} of wireless interference that addresses these questions in full. The approximations obtained are at most doubly logarithmic in the link length and rate diversity. Where previous bounds are known, this gives an exponential improvement (or better).

A key contribution is showing that the complex interference relationship of the physical model can be simplified, at a small cost, into a novel type of amenable conflict graphs. 
We also show that the approximation obtained is provably the best possible for any conflict graph formulation.
\end{abstract}

\thispagestyle{empty}

\end{titlepage}

\section{Introduction}

Graphs are ubiquitous structures that are used, among other things, for modelling
\emph{conflicts} between pairs of elements. Such conflicts arise naturally in resource allocation. An independent set in a graph corresponds to a subset of non-conflicting elements, while a vertex coloring of the graph implies a \emph{schedule} of the elements in groups of non-conflicting sets.
Such pairwise conflicts are though only the simplest form of \emph{constraints}.

An example of more general constraints on resource usage: ``at most two out of these three elements can be active simultaneously''. Such constraints are captured with  \emph{hypergraphs}, whose hyperedges correspond to not-all-active-simultaneously constraints. The concepts of independent sets and colorings carry over to hypergraphs as well. The downside of this generalization is that hypergraphs have proven to be much less amenable to efficient or effective solutions. They are also harder to reason about, with less powerful theoretic tools available.

This paper proposes a way to finesse the hardness of working with hypergraphs, by reducing them to graphs.
We form a \emph{sketch} of a given hypergraph that conservatively captures the essential constraints. The sketch is an ordinary graph
with the property that the solution of an optimization problem on the graph is also a valid solution in the hypergraph. Necessarily, the other direction need not hold exactly, but the big question is how much of a loss in precision is sacrificed by sketching. The obvious benefit of sketching is that the rich theory of graph algorithmics can be brought to bear, with commensurate conceptual simplifications.

The object of study in this work are certain geometrically-defined hypergraphs that capture interferences in wireless systems. Our main result is that they can be sketched at a low cost. This implies major improvements for a large family of such scheduling problems.

\mypar{Wireless scheduling}
The effective use of wireless networks revolves around utilizing fully all available diversity.
This can include power control, scheduling, routing, channel assignment and
transmission rate control on the communication links.
At the heart of this large space of optimization problems are certain fundamental problems, which either involve maximizing throughput within a time frame or minimizing the number of time slots. 

Consider the following prototypical problem, known as Max Weighted Independent Set of Links (\wcapacity):
We are given a set of links, each of which is a pair of sender and receiver nodes, and a positive weight associated with each link.
Underlying is a system of constraints that
stipulate which subsets of links can be simultaneously active due to the unavoidable \emph{interference} between links.
The objective is to find a maximum weight subset of links that can be simultaneously active.

To capture interference, the model of choice for analytic studies of wireless systems is the \emph{physical} or SINR (\emph{Signal to Interference and Noise Ratio}) model. 
Each node is located in a metric space and each active transmission incurs \emph{fractional} interference on every other link, that is a function of the relative positions of the nodes of the two links. A transmission is successful as long as the total interference from the other links does not exceed a given threshold.
This model is provably more accurate than binary (or graph-based) models.
It is not without its weaknesses in fully capturing the reality of wireless systems, which we will address later in the paper.
However, it is arguably the measuring stick with which we compare other models, and forms the basis of more refined models.

All the scheduling problems of interest here are NP-hard.
Our objective is to give efficient algorithms that provide good performance guarantees. When constant-approximations are out of reach, we seek slow-growing functions of the key parameters: $n$, the number of links, and $\Delta$, the
diversity in link lengths (i.e., the ratio between the length of the longest to the shortest link).
A secondary objective is to derive \emph{simple} algorithms based on \emph{local} rules, as such methods are most likely to be applicable or informative
in constrained system setting, e.g., distributed.

Our approach is to produce two graphs, $G_{lo}$ and $G_{hi}$, that \emph{sandwich} the input hypergraph $\cH$ in the following sense: every independent set of $G_{hi}$ is also an independent set of $\cH$, and every independent set of $\cH$ is also an independent set of $G_{lo}$. These graphs belong to a new class that generalizes the intersection graphs of disks, and they share the desirable properties of constant-approximability of (weighted) maximum independent set and graph coloring problems, among others.
For instance, to solve the {\wcapacity} problem on $\cH$, we simply run a weighted independent set algorithm on $G_{hi}$ and output the solution.

The ``price'' of the graph abstraction is given by the difference between the upper and the lower sandwich graphs. 
Technically, it is bounded by taking an independent set in $G_{lo}$ and
considering its chromatic number in $G_{hi}$.
This factor is either\footnote{All logarithms in this paper are base-2.} $O(\log^* \Delta)$ or $O(\log\log \Delta)$, depending on the setting.
We show that this is actually the best possible price that can be achieved with any conflict graph representation.

\subsection{Our Results}

We develop a general approximation framework that can tackle nearly all wireless scheduling problems, such as TDMA scheduling, joint routing and scheduling and others.
The problems
handled can additionally involve path or flow selection, multiple channels and radios, and packet scheduling.

The approximation factors are \emph{double-logarithmic} (in link and rate diversity) approximation for these problems, exponentially  improving the previously known logarithmic approximations, and, importantly, extending them to incorporate \emph{different fixed data rates and rate control}.

Our approach also finesses the task of selecting optimum power settings by using \emph{oblivious} power assignment, one that depends only on the properties of the link itself and not on other links. The performance bounds are however in comparison with the optimum solution that can use arbitrary power settings.

In the special case of fixed uniform rates (where all links require the same data rate), our approach yields an even better $O(\log^* \Delta)$-approximation, if we are willing to forego the advantage of oblivious power assignments. 
We show that this is actually the best possible, not only for our construction, but for any formulation involving conflict graph abstractions. 
The same holds for the double-logarithmic factor involving non-uniform data rates.

\mypar{Assumptions}
We make some undemanding assumptions about the settings. 
We assume that nodes can adjust their transmission power.

We assume that the networks are
\emph{interference-constrained}, in that interference, rather than the ambient noise, is the determining factor of proper
reception.  This assumption is common and is particularly natural in settings with rate control, since the impact of noise can always be made
negligible by avoiding the highest rates, losing only a small factor in performance.
We also assume that nodes are (arbitrarily) located in a doubling metric, which generalizes Euclidean space, allowing the modeling of some of  non-geometric effects seen in practice. We show that all of our assumptions are necessary (to obtain results of the form given here).
We have not attempted to minimize the constant factors involved in the analysis.

\mypar{Paper Organization} 
We first introduce our sandwiching technique in Sec.~\ref{s:sandwiching} 
and outline  the necessary properties of applicable problems.
We then describe in detail (in Sec.~\ref{s:problems}) a large class of scheduling problems and explain why our results apply to them.

The conflict graph construction is given in Sec.~\ref{s:conflict}, 
where we then proceed to bound in general terms the quality of the sandwiching attained. We also derive the key graph-theoretic properties that allow for constant approximability.

The most technical material is in Sec.~\ref{s:feas}, where we finally introduce the physical model of interference. The main effort is in showing that independent sets in the conflict graphs correspond to feasible sets of links (as per the hypergraph formulation).
This is shown separately for general fixed rates with oblivious power control, and for fixed uniform rates with arbitrary power control.

In Sec.~\ref{s:limitations}, we show that our formulations are best possible, both by showing that no better bounds can be achieved with our types of conflict graphs, and by arguing that every conflict graph formulation essentially matches one of our conflict graphs.
We also show that our assumptions are all necessary, including power control, metric space, and interference-limited setting.

Finally, we provide some context in Sec.~\ref{s:context}, 
first describing related work that did not fall purely under one of the problems studied (Sec. \ref{s:problems}). We then address the issue of strengths and weaknesses of models of interference.

\section{Sandwiching Hypergraphs with Graphs}
\label{s:sandwiching}

\mypar{Independence systems}
A \emph{hypergraph} $\cF = (V,\cE)$ consists of a collection $\cE$ of \emph{hyperedges},
which are subsets of a finite set $V$.
A graph is a hypergraph with edges only of size 2. 
In our context, the vertices of the hypergraph correspond to \emph{communication links} and the hyperedges encode constraints caused by interference: if a set of concurrently transmitting links contains one of the hyperedges, then some of the transmissions fail.

A subset of vertices is \emph{independent} if it contains no hyperedge.
The \emph{independence system} $\cI_\cF$ consists of all the independent sets in the hypergraph $\cF$.

\mypar{Sandwiching}
We seek a pair of graphs: a graph $\ghi$, that constrains the hypergraph from above, and $\glo$, that constrains it from below,
satisfying:
\[ \cI_{\ghi} \subseteq \cI_\cF \subseteq \cI_{\glo}\ . \]

Sandwiching by itself is trivial (using the empty and the complete graph)
but we seek graphs with not-too-different independence systems.
Specifically, the pair of graphs are a \emph{$\rho$-sandwich} if
  \[ \chi(\ghi[S]) \le \rho \cdot \chi(\glo[S])\ , \]
where $\chi(G)$ is the (vertex) chromatic number of $G$.
In other words, every independent set in $\glo$ can then be partitioned into at most $\rho$ independent sets in $\ghi$.
We refer to the smallest such $\rho$ as the \emph{tightness} of the sandwiching, 
which determines the quality of the sandwiching.

The graph $\glo$ used will simply consist of the 2-edges of $\cF$, namely the incompatible pairs of links: $E(\glo) = \{e \in \cE : |e|=2\}$.
We will generally omit the mention of $\glo$ and refer to $\ghi$ as the hypergraph sketch, as well as referring to the \emph{tightness} of $\ghi$.

The idea behind sandwiching is to obtain efficient approximations of an optimization problem involving independence constraints given by a hypergraph $\cF$ by simply solving the same problem with a modified independence system given by the graph $\ghi$. This always gives a feasible solution, and if the problem at hand is ``nice'' (as discussed below), then the tightness of sandwiching gives an upper bound on the efficiency of approximation.

\mypar{Properties of problems for which sandwiching applies}
Sandwiching can be applied to a wide variety of optimization problems that involve constraints in the form of a hypergraph $\cF$. The problems can, e.g., involve various other data outside of the scope of $\cF$. It suffices that three properties hold:
\begin{description}
\item[Monotonicity] If $\cF, \cF'$ are hypergraphs with
$\cI_\cF \subseteq \cI_{\cF'}$, then $OPT(\cF) \ge OPT(\cF')$ for a minimization problem, and $OPT(\cF) \le OPT(\cF')$ for a maximization problem, where $OPT$ is the optimum measure of the problem.
\item[Tightness] The increase (or decrease) in the objective function between the graphs in a $\rho$-sandwich is at most proportional to the tightness of the sandwiching, on every induced subgraph. Namely, 
  $OPT(\ghi[S])/OPT(\glo[S]) = O(\rho)$, for every $S \subseteq V$ (for minimization problems).
\item[Approximability] The problem admits a $c$-approximation algorithm on the class of sandwich graphs from which $\ghi$ is chosen, for a parameter $c$.
\end{description}

Given these properties, the strategy is simply to solve the problem at hand over the constraints given by the graph $\ghi$. We have a $c$-approximation for this restricted form, due to the approximability property, and the tightness and monotonicity properties ensure that, e.g., for a minimization problem, $OPT(\ghi) = O(\rho)\cdot OPT(\glo) = O(\rho)\cdot OPT(\cF)$. Hence, we have a $O(c\rho)$-approximation for the problem in $\cF$.

We show in Sec.~\ref{s:problems} how most wireless scheduling problems can be handled with this strategy. This approximation allows us to bring to bear the large body of theory of graph algorithms, simplifying both the exposition and the analysis. We also present several problems that do not fall under this framework, but can nevertheless be solved using sandwiching in a more customized manner.

\section{Wireless Scheduling Problems}
\label{s:problems}

In many wireless scheduling problems, the basic object of study is a set $L$ of $n$ (potential) communication links, where each link  $i\in L$ represents a single-hop communication request between two wireless nodes -- a sender node $s_i$ and a receiver node $r_i$.

Transmissions on links cause interference to other links. The transmission rate of a link that is scheduled in a given slot depends on its signal to interference ratio (SIR). We consider two kinds of scheduling problems. In \emph{fixed-rate} problems, every link $i$ has a fixed SIR threshold $\beta_i$, and the only requirement is that it achieve the rate associated with this threshold: a link is \emph{successful} if and only if it is scheduled so that its SIR is at least $\beta_i$. Such fixed thresholds give rise to a feasibility formulation that is described in terms of a hypergraph $\cF=(L,\cE)$ on the links: if $S \subseteq L$ is the set of links transmitting (in a given time/frequency slot),
then all the links in $S$ are successful if and only if $S \in \cI_\cF$. We say then that $S$ is a \emph{feasible} set of links.

We also consider problems involving \emph{rate control}. Here, the goal is not to achieve a fixed minimum rate, but to optimize some function of achieved data rates, e.g., the total rate over all links. Hence, this case is not described directly with the hypergraph formulation above, but we can reduce such problems to their fixed-rate variants (essentially) preserving the approximation factor.

The property of our conflict graphs that provides the approximability property is that they are \emph{$O(1)$-inductive independent}. More strongly, they are \emph{$O(1)$-simplicial}, as defined below (See Sec.~\ref{s:graphalgo} for proofs). 
A \emph{$k$-simplicial elimination order} is
one where the \emph{post-neighbors} of each vertex, or the neighbors appearing to its right,  can be covered with $k$ cliques.  A graph is $k$-simplicial
if it has a $k$-simplicial elimination order. In \emph{$k$-inductive independence} graphs, the set of post-neighbors of each vertex is only required to have independence number bounded by $k$ (hence, a $k$-simplicial graph is also $k$-inductive independent).
These graph classes have been well studied, and it is known that among others, vertex coloring and maximum weight independent set problems are $k$-approximable in $k$-inductive independent and $k$-simplicial graphs~\cite{akcoglu, kammertholey, yeborodin}.

\subsection{Fixed-Rate Problems}
\label{ss:fixedrate}

These problems can be classified as \emph{covering} or \emph{packing} problems, where in the former we seek to minimize the number of time slots, while in the latter to maximize a weighted feasible selection of links. Various other objectives might also apply, such as the sum of completion times (i.e.\ indices of time slots), that we do not address here. 

Monotonicity of all these problems is easy to check. They also have efficient approximations on our conflict graphs (and more generally on $O(1)$-inductive independent graphs). So we only need to demonstrate their tightness (for a given $\rho$-sandwich), when not obvious. In some special cases, we need an ad-hoc approach for obtaining the approximation.

It should also be mentioned that fixed-rate problems can be considered in two regimes: \emph{Uniform thresholds}, where the  thresholds $\beta_i$ are equal for all links, and \emph{general thresholds}, where there is no restriction. The only difference in our results concerning these two regimes is that the tightness $\rho$ of sandwiching is significantly better in the case of uniform thresholds. However, the analysis of the problems below does not depend on the particular regime, and assumes a general $\rho$-sandwich is given.

\subparagraph*{Max (Weight) Independent Set of Links ({\wcapacity})}
Find a feasible set of links of maximum cardinality or weight.

 A local-ratio algorithm gives constant-approximation in constant-simplicial graphs \cite{yeborodin}.

\subparagraph*{Admission Control}
The \emph{online} {\wcapacity}, or \emph{admission control} problem, is defined as follows: the links arrive one-by-one, and the algorithm is to irrevocably admit or reject the current link in the feasible set. The quality of a solution is evaluated via the \emph{competitive ratio}, that is, the ratio between the solution value obtained by the online algorithm and that of the optimum offline solution. 

It is known that deterministic online algorithms perform rather poorly, when compared with the offline optimum~\cite{FanghanelGHV13}. 
Hence,~\cite{GHKSV14} considers algorithms on \emph{stochastic input models}, such as the \emph{secretary model}, in which an  adversarial graph is presented in a random order, and the \emph{prophet-inequality model}, in which a random graph is presented in an adversarial order. They present expected constant-competitive ($O(\log n)$-competitive) algorithms for unweighted (weighted, resp.) variants of the problem on constant-inductive independent graphs. Applying this to $\ghi$ and using sandwiching, we obtain expected competitive ratios $O(\rho)$ and $O(\rho \log n)$, respectively, compared with the optimum offline solution in the hypergraph $\cF$.

\subparagraph*{(TDMA) Link Scheduling}
Partition the input set of links into the minimum number of feasible subsets. 

A simple \emph{first-fit} style greedy algorithm gives constant factor approximation to vertex coloring in constant-simplicial graphs \cite{yeborodin}.

\subparagraph*{Online Link Scheduling}
The online variant of Link Scheduling we consider is as follows. The links arrive one by one, in an online manner, and the algorithm should assign each arriving link to a time slot, so that the set of links in each slot is feasible, and the number of slots is minimal. Once a link is assigned to a slot, it cannot be moved to  another one, but its power level can be adjusted with newly arriving links, to reinforce feasibility.

In order to approximate the online scheduling problem, we simply apply an online vertex coloring algorithm to the graph $\ghi$. A graph $G$ is \emph{$d$-inductive} if there is an ordering of the vertices, such that each vertex has at most $d$ post-neighbors in the ordering. It is well known that a simple greedy online algorithm colors $d$-inductive graphs using $O(d\log{n})$ colors~\cite{iranionline}, where $n$ is the number of vertices.  
It is a simple observation that every constant-inductive independent graph $G$ is $O(\chi(G))$-inductive. Hence, we have an algorithm that colors $\ghi$ with $O(\chi(\ghi)\log n)$ colors. By sandwiching, $\chi(\ghi)\le \rho \chi(\glo)$, implying that the obtained algorithm is $O(\rho \log n)$-competitive, compared with the optimum offline solution in the hypergraph $\cF$.

\subparagraph*{Multi-Channel Selection}
Given a natural number $c$ -- the number of channels -- select a maximum number (or weight) of links that can be partitioned into $c$ feasible subsets (a subset for each channel).

There is a constant-factor approximation algorithm for constant-simplicial graphs \cite{yeborodin}.

\subparagraph*{Fractional Scheduling}
In this fractional variant of Link Scheduling, we are additionally given a real-valued demand $d(i)$ on each link $i$, indicating the amount of time that each link needs to be scheduled.
A \emph{fractional schedule} of the links is a collection of feasible sets with rational values $\cS=\{(I_k,t_k) : k=1,2\dots,q\}\subseteq \cI_{\cF}\times\mathbb{R}_+$, where $\cI_{\cF}$ is the set of all feasible subsets of $L$. The sum $\sum_{k=1}^q{t_k}$ is the \emph{length} of the schedule $\cS$. The \emph{link capacity vector} $c_{\cS}:L\rightarrow \mathbb{R}_+$ associated with the schedule $\cS$ is given by
$c_{\cS}(i) = \sum_{(I,t)\in \cS: I\ni i}t$, indicating how much scheduling time the link gets.

The \emph{fractional scheduling problem} is a covering problem, where given a demand vector $d$, the goal is to compute a minimum length schedule that serves the demands, namely,  for each link $i\in L$,
$c_{\cS}(i)\ge d(i)$.

A greedy algorithm presented in~\cite{wan13} achieves constant-approximation on constant inductive independent graphs. 

\subparagraph*{Joint Routing and Scheduling}
 Consider a set of source-destination node pairs (multihop communication requests) $(u_i,v_i)$, $i=1,2,\dots,p,$ with associated weights/utilities $\omega_i>0$. The nodes are located in a multihop network given by a directed graph $G$, where the \emph{edges} of the graph are the transmission links. Let $\cP_i$ denote the set of directed $(u_i,v_i)$ paths in $G$ and let $\cP=\cup_i \cP_i$. 

A \emph{path flow} for the given set of requests is a set $F=\{(P_k,\delta_k): k=1,2,\dots\}\subseteq \cP \times \mathbb{R}_+$.  The \emph{link flow vector} $f_{F}$ corresponding to path flow $F$, with
$f_{F}(i)=\sum_{(P,\delta)\in F: P\ni i}{\delta},$
gives the flow along each link $i$.

The \emph{multiflow routing and scheduling problem} is a covering problem, where given source-destination pairs with associated utilities, the goal is to find a path flow $F$ together with a fractional link schedule $\cS$ of length $1$, such that\footnote{Essentially, the schedule here gives a probability distribution over the feasible sets of links.} for each link $i$, the link flow is at most the link capacity provided by the schedule, $f_{F}(i) \le c_{\cS}(i)$, and the \emph{flow value}
\[
W=\sum_{i=1}^p \omega_i\cdot \sum_{(P_k,\delta_k)\in F, P_k\in \cP_i}{\delta_k}
\]
 is maximized.

A constant-approximation algorithm of~\cite{wan14} (the result holds with unit utilities) for constant-inductive independent graphs applies here.
It should also be noted that the fractional scheduling and routing and scheduling problems can be reduced to the {\wcapacity} problem using linear programming techniques (described e.g.\ in~\cite{jansen03}), as shown in~\cite{wan09}. We will further discuss this in Sec.~\ref{ss:ratecontrol}.

Let us verify that this problem satisfies the tightness property. Consider a feasible solution in $\glo$ that consists of a path flow $F=\{(P_k,\delta_k): k=1,2,\dots\}$ and a schedule $\cS=\{(I_k,t_k) : k=1,2,\dots\}$ of length $\sum_{k\ge 1}t_k=1$, such that $f_{F}(i) \le c_{\cS}(i)$. By the sandwiching property, the schedule $\cS$ can be refined into a schedule $\cS'=\{(I_k^s,t_k)\}_{k,s}$ in $\ghi$, where $\cS'$ serves the same demand vector as $\cS$, and $\cS'$ has length at most $\rho$ times the length of $\cS$. We then scale the refined schedule to have length 1, so that the scaled path flow $F'=\{(P_k,\delta_k/\rho): k=1,2,\dots\}$ together with the new schedule will be feasible in $\ghi$, as all link demands will be served. Clearly, the value of $F'$ is at least a $1/\rho$ fraction of the value of $F$, so we achieve a tightness of $\rho$.

\subparagraph*{Multi-Channel Multi-Antenna Extensions}
 All the problems above can be naturally generalized to the case when there are multiple channels (e.g.\ frequency bands) available and moreover, wireless nodes are equipped with multiple antennas and can operate in different channels simultaneously (MC-MA.)
Each node $u$ is equipped with $a(u)$ antennas numbered from $1$ to $a(u)$ and can (only) use a subset $\cC(u)$ of  channels. 

For each link $i = (s_i, r_i)$, we form a collection of 
$a(s_i)a(r_i)|\cC(s_i)\cap \cC(r_i)|$ \emph{virtual} links, 
that correspond to each selection of an antenna of the sender node $s_i$, an antenna of receiver node $r_i$ and a channel $c\in \cC(s_i)\cap \cC(r_i)$ available to both nodes.  We call link $i$ \emph{the original} of its virtual links. 
A set of virtual links $S$ is feasible in MC-MA if and only if no two links in $S$ share an antenna (i.e., they do not use the same antenna of the same node), and the set of originals of links in $S$ using each channel is feasible (in $\cF$). We show that the conflict graphs $\ghi$ and $\glo$ can be extended to this setting, preserving their properties.

Let $L$ denote the set of virtual links and $L_o$ the corresponding originals. We define the conflict graphs $\ghim(L)$ and $\glom(L)$ that have a node for each virtual link, with two virtual links adjacent if at least one of the following holds: 1. they share an antenna, or 2. they share a channel \emph{and} their originals are adjacent in $\ghi(L_o)$ or respectively in $\glo(L_o)$, i.e., in the single channel setting. In particular, the replicas of the same original link form an independent set in both graphs.

We prove that if $\ghi$ is $k$-simplicial, then $\ghim$ is $k+2$-simplicial. The other properties follow by similar arguments. To this end, consider a virtual link $i\in L$, and let us see which links are in the neighborhood of $i$. The neighborhood of $i$ can be partitioned into three sets: 1. The virtual links that share the channel with $i$, denoted $O$, 2. The links that use the sender antenna of  $i$, denoted $S$, and 3. The links that use the receiver antenna of $i$, denoted $R$. Note that $O$ consists of replicas of \emph{distinct} links in $L_o$, which are all adjacent with the original of $i$ in $\ghi$. Also, note that $S$ and $R$ form cliques in $\ghim$. It is now easy to see that $\ghim$ is $k+2$-simplicial, where the simplicial ordering is induced by the simplicial ordering of $\ghi$.

\subparagraph*{Spectrum Auctions With Sub-Modular Valuations}
The \emph{spectrum auction problem} is a packing problem that can be considered a generalization of {\wcapacity}, where there are multiple channels and a not-necessarily-additive weight function. 

Given a set $L$ of links, a natural number $c$ (number of available channels) and a valuation function $\omega:L\times 2^{[c]}\rightarrow \mathbb{N}$, find a \emph{feasible} allocation $A:L\rightarrow 2^{[c]}$ that maximizes the sum of valuations $\omega(A)=\sum_{i\in L}{\omega_{i,A(i)}}$, where $[c]={1,2,\dots,c}$. Note that each feasible allocation is a collection of $c$ feasible sets, each corresponding to a channel.
Note also that the problem is reduced to solving a number of {\wcapacity} problems when the valuation function is additive, i.e.\ $\omega_{i,T}=\sum_{j\in T}\omega_{i,j}$, leading to a $O(\rho)$-approximation.

In the more general case when the valuation function $\omega_{i,T}$ is a \emph{submodular} function of $T$ for each link $i$, i.e.\ for any sets $T,T'$ of channels, $\omega_{i,T\cup T'} + \omega_{i,T\cap T'}\le \omega_{i,T} + \omega_{i,T'}$, randomized algorithms presented in~\cite{HoeferK15} give constant-factor approximation for constant-inductive independent graphs and $O(\log n)$-approximation for the physical model, in expectation. These approximations hold for a particular kind of submodular functions called \emph{matroid rank sum functions}. Thus, in order to obtain an (expected) $O(\rho)$-approximation for matroid rank sum functions, we only need to verify the tightness property. 

First, note that any non-negative submodular function $f$ is subadditive, i.e.\ for each set $S$, $f(S)\le \sum_{e\in S} f(e)$. Consider a feasible allocation $S_1,S_2,\cdots,S_c$ in $\glo$. Using sandwiching, we can split each $S_t$ into $\rho$ independent sets $S_t^1,S_t^2,\cdots,S_t^\rho$ in $\ghi$ (where some of the subsets may be empty). Consider (at most) $\rho$ tentative allocations $\{S_1^j,S_2^j,\cdots,S_c^j\}$ for $j=1,2,\cdots,\rho$ and consider the sum of total valuations of these allocations. Let $i$ be any fixed link. In each of the obtained allocations, link $i$ gets a subset of channels and the subsets corresponding to different allocations are disjoint and sum up to the  set of channels allocated to $i$ in the original allocation. This observation and the fact that the valuation function for each link is subadditive imply that the sum of total valuations is at least the total valuation of the original allocation. Since there are at most $\rho$ refined valuations, this implies that the best one of them gives total valuation at most $\rho$ times that of the original allocation.

\subparagraph*{Spectrum Auctions with General Valuations}
When the valuation functions $\omega_{i,T}$ are unrestricted, our framework may not be applied directly, as we cannot guarantee the tightness property. However, in this case we can take advantage of a particular solution proposed in \cite{HoeferKV14}, where a linear programming approach is developed. Using this approach an $O(\sqrt{c})$-approximation is obtained for constant inductive independent graphs, where $c$ is, as before, the number of channels.

 For a vertex $v$ in a $k$-inductive independent graph $G$, let $N^+_G(v)$ denote the set of post-neighbors in the inductive independence order. The linear program (which is a relaxation of the corresponding integer linear program (ILP)) for $k$-inductive independent graphs presented in~\cite{HoeferKV14} is as follows.
\begin{align}
\nonumber \text{Maximize } &\sum_{v\in V}{\sum_{T\subseteq[c]}{\omega_{v,T}x_{v,T}}}   &&\\
 \nonumber \text{s.t. } &\sum_{u\in N^+_G(v)}{\sum_{T\subseteq [c],T\ni j}{x_{u,T}}} \le k && v\in V, j\in [c]\\
\label{E:ilp} & \sum_{T\subseteq [c]}{x_{v,T}} \le 1 && v\in V\\
\nonumber & x_{v,T}\ge 0 && v\in V,T\subseteq[c]
\end{align}

The first constraint corresponds to $k$-inductive independence: The number of post-neighbors of a vertex $v$ that are assigned the same channel $j$ must be bounded by $k$. The second constraint states that each vertex is assigned a single set of channels.

An algorithm based on randomized rounding of the linear program solution  is presented in~\cite{HoeferKV14}, giving $O(k\sqrt{c})=O(\sqrt{c})$-approximate solution in expectation. Again, the problem with this solution is that in the absence of tightness property, as we do not know how the ILP solution compares with the optimal solution in $\glo$. However, we can ``plant'' the tightness property in the ILP, as follows. The key observation is that the only constraint that really depends on the underlying graph is the inductive independence constraint, so by simply replacing the right-hand side of the constraint with $k'=k\cdot \rho$, we obtain, due to sandwiching, that every solution in $\glo$ is a feasible solution in the ILP\footnote{Here we also use the fact that $\ghi$ and $\glo$ have the same inductive independence order.},  even though  the ILP is formulated in terms of $\ghi$. This means that the ILP optimum is a lower bound on the optimum in $\glo$. Hence, the randomized rounding algorithm of~\cite{HoeferKV14} gives us a $O(\rho\sqrt{c})$-approximation of the optimum in $\cF$.

A similar approach can be used to obtain an expected $O(\rho)$-approximation for another special case considered in~\cite{HoeferK15}, when the valuation function for each link is \emph{symmetric}, i.e.\ the valuation depends only on the number of channels rather than on specific subsets: for each link $i$ and subsets $T,T'$ of channels, $\omega_{i,T}=\omega_{i,T'}$ if $|T|=|T'|$.

\subsection{Rate Control and Scheduling}
\label{ss:ratecontrol}

Most of the fixed-rate problems also have variants where choosing the data rates is part of the problem.
We describe here how these problems can be reduced to fixed-rate problems with minimal overhead.
Again, the approximability of the problems in the physical model depends linearly on the tightness of the sandwiching.\footnote{In general, the tightness parameter depends on rates, and since we don't have fixed rates in this case, tightness will depend on max/min rates.}
It is important to stress that the reduction is done to fixed-rate problems with \emph{non-uniform} thresholds, which means that the approximation guarantees for non-uniform thresholds apply here.

\subparagraph*{{\capswcapacity} with Rate Control}
 By Shannon's theorem, given a set $S$ of links simultaneously transmitting in the same channel, the transmission data (bit-)rate $r(S,i)$ of a link $i$ is a non-decreasing function of the SIR, $SIR(S,i)$, of link $i$ (with other parameters, e.g.\ frequency, fixed). Thus, we consider the {\wcapacity} problem where each link $i$ has an associated non-decreasing \emph{utility function} $u^i:\mathbb{R}_+\rightarrow \mathbb{R}_+$, and the weight of link $i$ is the value of $u^i$ at $SIR(S,i)$ if link $i$ is selected in the set, and $0$ otherwise. The goal is, given the links with utility functions,  to find a subset $S$ that maximizes the total utility $\sum_{i\in S} u^i(r(S,i))$. We assume that $u^i(r(S,i))=0$ if $SIR(S,i)<1$, namely when the signal is weaker than the interference.

An $O(\log n)$-approximation for this variant of {\wcapacity} was obtained in~\cite{KesselheimESA12}. We show, by reducing the problem to {\wcapacity} in a modified fixed-rate instance, that this ratio can be replaced with $O(\rho)$, where $\rho$ is the tightness of the fixed-rate instance.

Let us fix a utility function $u$. First, assume that the possible set of weights for each link is a discrete set $u_{min}=u_1<u_2<\cdots<u_{t}=u_{max}$. Then, we can replace each link $i$ with $t$ copies  $i_1,i_2,\cdots,i_t$ with different thresholds and fixed weights, where $\omega_{i_k}=u_k$ and $\beta_{i_k}=\min\{x: u^{i_k}(x) \ge u_k\}$, but $\omega_{i_k}=0$ if $\beta_{i_k}< 1$ (the latter is justified by our assumption that $SIR<1 \Rightarrow u=0$). Now, the problem becomes a {\wcapacity} problem for the modified instance $L'$ with link replicas and fixed weights. Observe that no feasible set in $L'$ contains more than a single copy of the same link, \footnote{This follows from the definition of the physical model (Sec.~\ref{s:model}), the assumption that $\beta_i\ge 1$ for each link $i$, and that the copies occupy the same geometric place.} implying that  each feasible set of the modified instance corresponds to a feasible set of the original instance, with an obvious transformation. 

For the case when the number of possible utility values is too large or the set is continuous, a standard trick can be applied. Let $u^i_{max},u^i_{min}$ be  the minimum and maximum possible utility values for the given link $i$. The  modified instance $L'$ is constructed by replacing each link $i$ with $O(\log u^i_{max}/u^i_{min})$  copies $i_1,i_2,\dots$ of itself and assigning each replica $i_k$ weight $\omega_k=2^{k-1}$ and threshold $\beta_k=\min\{x : 2^{k-1} \le u^i(x) \le 2^{k}\}$ if $\beta_k\ge 1$ and let $\omega_k=0$ otherwise. 

If the value $\log u^i_{max}/u^i_{min}$ is still too large, it may be inefficient to have $O(\log u^i_{max}/u^i_{min})$ copies for each link. It is another standard observation that only the last $O(\log n)$ copies of each link really matter, as restricting to only those links  degrades the approximation by at most a factor 2.

\subparagraph*{Fractional Scheduling with Rate Control}
In this formulation, we redefine a fractional schedule to be  a set $\cS=\{(I_k,t_k) : k=1,2\dots,q\}\subseteq 2^L\times\mathbb{R}_+$, namely, $I_k$ are arbitrary subsets rather than independent ones.  We redefine the link capacity vector ${\hat c}_{\cS}$ to incorporate the data rates as follows:
\begin{equation}\label{E:ratecap}
{\hat c}_{\cS}(i) = \sum_{(I,t)\in \cS: I\ni i}t\cdot r(I,i).
\end{equation}
 The \emph{fractional scheduling with rate control} problem is to find a  minimum length schedule $\cS$ that serves a given demand vector $d$, namely, such that  for each link $i\in L$,
$
{\hat c}_{\cS}(i)\ge d(i).
$

The problem can be formulated as an exponential size linear program $LP_1$, as follows.
\begin{align}
\nonumber \text{Minimize } &\sum_{I\subseteq L}{t_I}   &&\\
\nonumber \text{s.t } & \sum_{I\subseteq L : I\ni i}t_I \cdot r(I,i) \ge d(i) && \forall i\in L\\
\nonumber & t_I\ge 0 && \forall I\subseteq L
\end{align}

The dual program $LP_2$ is then:
\begin{align}
\nonumber \text{Maximize } &\sum_{i\in L}{d(i) y_i}  &&\\
\nonumber \text{s.t. } & \sum_{i\in I}y_i\cdot  r(I,i) \ge 1 && \forall I\subseteq L\\
\nonumber & y_i\ge 0 && \forall i\in L
\end{align}

As~\cite[Thm. 5.1]{jansen03} states, if there is an approximation algorithm that finds a set $\hat I$ such that $\sum_{i\in \hat I}y_i r(\hat I,i)\ge \frac{1}{a}\max_{I\subseteq L}\sum_{i\in I}y_i r(I,i)$, then there is an $a$-approximation algorithm for $LP_1$, where the former algorithm acts as an approximate \emph{separation oracle} for $LP_1$. This auxiliary problem is simply a special case of  {\wcapacity} with rate control. Thus, there is an approximation preserving reduction from the fractional scheduling with rate control to {\wcapacity} with rate control.

\subparagraph*{Routing, Scheduling and Rate Control}
The rate-control variant of the routing and scheduling problem is formulated in the same way as for the fixed rate setting, with only the capacity constraints modified to involve the modified link capacity vector ${\hat c}_{\cS}$ of (\ref{E:ratecap}) incorporating the data rates on the links, instead of $c_{\cS}$.

This problem can also be reduced to {\wcapacity} with rate control, using similar methods as for the fractional scheduling problem. The reduction is nearly identical to the reduction of fixed rate versions of these problems to {\wcapacity}, presented in~\cite[Thm. 4.1]{wan09}.

\section{Conflict Graphs}
\label{s:conflict}

Consider a set $L$ of links, whose nodes 
are represented as points in a metric space with distance function $d$.
We denote $d_{ij}=d(s_i,r_j)$\label{G:asymdistance} and denote by $l_i=d(s_i,r_i)$\label{G:li}
the \emph{length} of link $i$, where $s_i$ ($r_i$) denotes the sender node (receiver node, resp.) of link $i$. Let further $d(i,j)$ denote the minimum distance between the nodes of links $i$ and $j$.

Each link $i$ has an associated \emph{sensitivity} $\s_i$, which indicates how sensitive it is to interference. This depends linearly on the strength of the transmission on the link, which depends on the length of the link, but can also depend on the coding. Higher data rates mean higher sensitivity.
For technical reasons, we shall require that $\s_i \ge 4l_i$; we show in Sec.~\ref{s:model} why this is apropos.

  Let $\Ds(S)=\max_{i,j\in S}\{\s_i/\s_j\}$ denote the \emph{sensitivity diversity} of a set $S\subseteq L$ of links.  

The conflict graphs are parameterized by a positive function $f$.
The graph $G_f = (L,E)$ is defined by
\begin{equation}
 (i,j) \in E \quad \Leftrightarrow \quad  d_{ij}d_{ji} < \s_i\s_j f\left(\s_{max}/\s_{min}\right)\ ,
\label{eq:gengraphdef}
\end{equation}
where $\s_{min}=\min\{\s_i,\s_j\},\s_{max}=\max\{\s_i,\s_j\}$.

When the sensitivity of links is proportional to the length, i.e.,\ $\s_i\sim l_i$ for all $i\in L$ (the ``uniform thresholds'' case considered below), the graph definition can be simplified to:
\begin{equation}
 (i,j) \in E \quad \Leftrightarrow \quad  d(i,j) < \s_{min} f\left(\s_{max}/\s_{min}\right)\ ,
\label{eq:unifgraphdef}
\end{equation}

We will generally assume that $f$ is \emph{sub-linear}, i.e., $f(x)=o(x)$.
We will choose the graph $\glo$ to be $G_1$ (i.e., $f(x)\equiv 1$), and $\ghi$ to be $G_f$ for an appropriate non-decreasing sublinear function $f$, depending on the setting, as discussed in Sec.~\ref{s:feas}. If follows easily from the properties of the physical model (Sec.~\ref{s:model}) that $\glo$ is precisely the set of $2$-edges of the hypergraph corresponding to the physical model.

We next show the key properties of these conflict graphs: their tightness and efficiency of (approximate) computation. We first need to discuss the metric under consideration.

\emph{Metric}. We will assume that the metric in which the nodes are
located shares some of the aspects of the Euclidean
plane. Specifically, the number of unit balls that can fit without
overlap in an $R$-ball is bounded by a polynomial in $R$.  Formally,
we consider metric spaces of \emph{bounded doubling dimension}, where
the doubling dimension $m$ of the metric space is the infimum
of all numbers $\delta > 0$ such that for every $\epsilon \in (0,1]$,
every ball of radius $r>0$ has at most $C\epsilon^{-\delta}$ points of
mutual distance at least $\epsilon r$, where $C\geq 1$ is an absolute
constant.  This generalizes Euclidean spaces, as the $m$-dimensional
Euclidean space has doubling dimension $m$~\cite{heinonen}.

\subsection{Tightness}

We begin by bounding the number of independent sets in $\ghi=G_f$ that are necessary to cover a feasible set. We show that this number is $O(f^*(\Ds(S)))$ for any feasible set $S$, where $f^*$, the \emph{iterated} $f$, is defined for every sub-linear function, as follows.

For each integer $c\geq 1$, the function $f^{(c)}(x)$ is defined inductively by: $f^{(1)}(x)=f(x)$ and $f^{(c)}(x)=f(f^{(c-1)}(x))$\label{G:frepeated}. Let $x_0=\sup\{x\geq 1, f(x) \ge x\} +1$; such a point exists for every $f(x)=o(x)$. The function $f^*(x)$\label{G:fstar} is defined by:
$
f^*(x)=\arg\min_c\{f^{(c)}(x)\le x_0\}
$ for arguments $x> x_0$, and $f^*(x)=1$ for the rest.
Note that for a function $f(x)=\gamma x^\delta$ with constants $\gamma>0$ and $\delta\in (0,1)$, $f^*(\Ds)=\Theta(\log{\log{\Ds}})$, while for $f(x)=\gamma \log^{t} x$ with constants $\gamma>0$ and $t\ge 1$, $f^*(\Ds)=\Theta(\log^*{\Ds})$. Those will be the functions we are most interested in.

\begin{theorem}\label{T:sandwich}
Our conflict graphs are $O(f^*(\Ds(S)))$-tight, for any non-decreasing sub-linear function $f$.
That is, if $S$ is independent in $\glo$, then $\chi(\ghi[S]) = O(f^*(\Ds(S)))$.
\end{theorem}

Fix a non-decreasing sub-linear function $f$.
The proof requires the following three lemmas, which encapsulate the technicalities of dealing with our conflict graphs.
\begin{lemma}\label{L:adjacent}
Let links $i,j$, $\s_i\le \s_j$, be adjacent in $\ghi$. Then, $d_{ij} + d_{ji} \le l_i + l_j + 2\sqrt{\s_i\s_j f(\s_j/\s_i)}$.
\end{lemma}
\begin{proof}
Denote $m_{ij}=\min(d_{ji}, d_{ij})$ and $m'_{ij} = \max(d_{ji}, d_{ij})$.
Since $i$ and $j$ are adjacent, $m_{ij}m'_{ji}\le \s_i\s_jf(\s_j/\s_i)$, 
which implies that $m_{ij} \le \sqrt{\s_i\s_j f(\s_j/\s_i)}$.
By the triangular inequality, $m'_{ij} \le l_i + l_j + m_{ij}$.
\end{proof}

\begin{lemma}\label{L:triangles}
Let $i,j,k$ be links and $c>0$ be a number, such that $c\cdot \s_i \le \s_j\le \s_k$, $f(x)\le x/22$ for all $x\ge c$, and $i$ is adjacent to  both $j$ and $k$ in $\ghi$. Then,
$
d_{jk}d_{kj} < 3\s_i\s_kf(\s_k/\s_i) + (2/3)\s_j\s_k.
$
\end{lemma}
\begin{proof}
The triangle inequality and the assumptions $\s_t\ge 4l_t$ (for all $t$) and $\s_j\ge \s_i$ imply that 
\begin{align*}
d_{jk}&\le \min(d_{ik} + l_i+d_{ji}, d_{ik} + l_j+d_{ij})\le \s_j/4  + m_{ij}+ d_{ik},\\
d_{kj}&\le \min(d_{ki} + l_i+d_{ij}, d_{ki} + l_j+d_{ji}) \le \s_j/4  + m_{ij} + d_{ki},
\end{align*}
where we denote $m_{ij}=\min(d_{ji}, d_{ij})$. Multiplying these inequalities gives us:
\begin{equation}\label{E:trianglesmain}
d_{jk}d_{kj} \le \s_j^2/16 + m_{ij}^2 + \s_{j}m_{ij}/2 + (\s_j/4 + m_{ij})(d_{ik} + d_{ki}) + d_{ik}d_{ki}.
\end{equation}
Denote $g_{j} = \s_i \s_j f(\s_j/\s_i)$ and $g_k=\s_i \s_k f(\s_k\s_i)$. 
Since links $j,k$ are adjacent to link $i$, we have that $m_{ij}\le \sqrt{d_{ij} d_{ji}} \le \sqrt{g_j}$ and $d_{ik}d_{ki}\le g_k$, and using Lemma~\ref{L:adjacent}, we also have that $d_{ik} + d_{ki}\le \s_i/4 + \s_k/4 + 2\sqrt{g_k}$. By plugging these bounds in (\ref{E:trianglesmain}) and simplifying, we obtain
\begin{align*}
d_{jk}d_{kj} &< \frac{3\s_j\s_k}{16} + g_j + \frac{\s_j}{2}\sqrt{g_j} 
+ \s_j\sqrt{g_k} + \frac{\s_k}{2}\sqrt{g_j} + 2\sqrt{g_jg_k} + g_k\\
&\le \frac{3\s_j\s_k}{16} + \frac{\s_j^2}{22} + \frac{\s_j^2}{2\sqrt{22}} + \frac{\s_j\s_k}{\sqrt{22}} + \frac{\s_j\s_k}{2\sqrt{22}} 
+ 2g_k + g_k\\
&< 3g_k + (2/3)\s_j\s_k,
\end{align*}
where to obtain the second inequality, we used the assumption that $f(\s_j/\s_i)\le \s_j/(22\s_i)$ and $f(\s_k/\s_i)\le \s_k/(22\s_i)$, and the inequality $g_j\le g_k$, which follows from the assumption that $f$ is non-decreasing and that $\s_k\ge \s_j$. 
\end{proof}

\begin{lemma}\label{P:simpleset} 
Let $c>0$ be a constant, $i$ be a link, and $S$ be a set of neighbors of $i$ in $\ghi$ such that $\s_i \le \s_j\le c\cdot \s_i$ holds for all $j\in S$. Then, $S$ can be partitioned into $O(1)$ cliques in $\glo$.
\end{lemma}
\begin{proof}
 Observe that any two links, whose receivers are within distance $3\s_i/4$, are adjacent in $\glo$. Indeed, if $j,k$ are such that $d(r_j,r_k)\le 3\s_i/4$, then by the triangle inequality, $d_{jk}\le l_j+3\s_i/4 \le \s_j$ (recall that $\s_j\ge 4l_j$), and similarly, $d_{kj}\le \s_k$, implying that $d_{jk}d_{kj}\le\s_j\s_k$. Now, consider the following partitioning of $S$ into cliques: 1. Start with an arbitrary link $i$, and let $K_i$ be the set of all links $j\in S$ with $d(r_j,r_i)\le 3\s_i/8$. 2. Remove $K_i$ from $S$ and repeat, until $S$ is empty. 

This procedure partitions $S$ into subsets $K_{i_t}$ indexed by links $i_t$. By the discussion above, each subset $K_{i_t}$ is a clique in $\glo$. Moreover, by construction, the sender nodes of $i_t$ and $i_{t'}$ for $t\neq t'$ are at mutual distance at least $3\s_i/8$. On the other hand, as assumed, each link $j$ in $S$ satisfies $d_{ji} d_{ij} \le \s_i \s_j f(\s_j/\s_i) \le c f(c)\s_i^2 $, so $d(i,j) \le \sqrt{c \cdot f(c)} \s_i$.
Then, it is easy to see that all sender nodes $r_j$ of links $j\in S$ are located within a ball of radius $c' \s_i$, with $r_i$ as center, where $c'=c+1+\sqrt{cf(c)}$.  Hence, the number of cliques obtained is at most $(8c'/3)^m=O(1)$, where $m$ is the doubling dimension of the space.
\end{proof}

\begin{proof}[Proof of Theorem~\ref{T:sandwich}]
Recall that a graph is $d$-inductive (or $d$-degenerate) if there is an ordering of the vertices so that each vertex has at most $d$ post-neighbors. It is well known that a greedy algorithm uses at most $d+1$ colors on $d$-inductive graphs.
Thus, it suffices for us to show that each independent set $S$ in $\glo$ induces a $O(f^*(\Delta))$-\emph{inductive} subgraph in $\ghi$, with respect to a non-decreasing order of links by sensitivity. Namely, for every link $i\in S$, the number of links $j\in S$ with $\s_j\ge \s_i$ that are adjacent to $i$ is in $O(f^*(\Delta(S)))$. 

To show inductiveness, consider a link $i\in S$, and let $S^+$ denote the set of links $j\in S$ with $\s_j\ge \s_i$ that are adjacent to $i$ in $\ghi$. Further, partition $S^+$ into subsets $S_{\ge c}^+$ and $S_{<c}^+$, containing the links $j$ with sensitivity greater (respectively, less) than $c\s_i$, where $c>0$ is a constant described in Lemma~\ref{L:triangles}. 
	Lemma~\ref{P:simpleset} implies that $|S_{<c}^+|=O(1)$, so it remains to show that $|S_{\ge c}^+|=O(f^*(\Delta(S)))$.
	
	Let $j,k\in S^+_{\ge c}$. Lemma~\ref{L:triangles} implies that $d_{jk}d_{kj} < 3\s_i\s_kf(\s_k/\s_i) + (2/3)\s_j\s_k$. On the other hand, we have $d_{jk}d_{kj}>\s_j\s_k$, since $j$ and $k$ are not adjacent in $\glo$. Combining, we obtain that $(1/3) \s_i \s_k < 3\s_i \s_k f(\s_k/\s_i)$, which leads to 
$\lambda_j < 9f(\lambda_k)$, using notation $\lambda_t=\s_t/\s_i$. Assume, w.l.o.g., that $S^+_{\ge c}=\{1,2,\dots,h\}$, and $\s_j\le \s_k$ for $j<k$. Then, denoting $g(x)\equiv 9f(x)$ we have:
	\[
	c\le \lambda_1 <g(\lambda_2)<g(g(\lambda_3))<\dots<g^{(h-1)}(\lambda_h),
	\]
	which, together with the assumption that $f(x)<x$ for all $x\ge c$ (see Lemma~\ref{L:triangles}) implies that $h-1\le g^*(\lambda_h)=O(f^*(\Delta(S)))$.
\end{proof}

\subsection{Algorithmic Properties of the Conflict Graphs}\label{s:graphalgo}

Computability of our conflict graph construction is demonstrated through the notion of \emph{$k$-simplicial graphs}, which generalize chordal graphs. In particular, we show that every conflict graph $G_f$ has a constant-simplicial elimination ordering, where the set of post-neighbors of every vertex can be covered with $O(1)$ cliques.

A function $f$ is \emph{strongly sublinear} if for each constant $c\ge 1$, there is a constant $c'$ such that $cf(x)/x\le f(y)/y$ for all $x,y\ge 1$ with $x\ge c'y$. 
For example, the functions $f(x)=x^{\delta}$, $\delta<1$, and $f(x)=\log{x}$ are strongly sublinear.

\begin{theorem}\label{T:inductiveindep}
Let $f$ be a strongly sublinear function with $f(x)\ge 1$ for all $x\ge 1$. For every set $L$, the graph $G_f(L)$  is $O(1)$-simplicial.
The corresponding ordering is given by non-decreasing sensitivity.
\end{theorem}

\begin{proof}
Order the links in non-decreasing sensitivity (ties broken arbitrarily).
Fix  a link $i\in L$, and let $T$ be its neighbors in $G_f$ with $\s_j\ge\s_i$, for all $j\in T$. 
We shall show that the links in $T$ of sensitivity greater than $c \s_i$ form a single clique, for some constant $c=c_f$ (specified below).
The links $j\in T$ of sensitivity $\s_j\le c \s_i$ can be covered with constant number of cliques, by Lemma~\ref{P:simpleset}.

Let $j,k$ be two links in $T$, with $\s_k\ge \s_j\ge c\s_i$. It suffices to show that $j$ and $k$ are adjacent in $G_f$.
By Lemma~\ref{L:triangles}, we have that
$
d_{jk}d_{kj} \le 3\s_i\s_kf(\s_k/\s_i) + (2/3)\s_j\s_k.
$
Since $f(x)$ is strongly sublinear we can choose constant $c$ such that $f(x)/x\le (1/9)f(y)/y$ for all $x,y$ with $x\ge c y$. Hence, $f(\s_k/\s_i)/(\s_k/\s_i) \le (1/9)f(\s_k/\s_j)/(\s_k/\s_j)$, provided that $\s_k/\s_i\ge c\s_k/\s_j$, i.e., $\s_j\ge c\s_i$. Thus, 
\[
d_{jk}d_{kj} \le (1/3)\s_j\s_kf(\s_k/\s_j) + (2/3)\s_j\s_k\le \s_j\s_kf(\s_k/\s_j),
\]
since $f(\s_k/\s_j)\ge 1$. Hence $j$ and $k$ are adjacent in $G_f$, as claimed.
\end{proof}

\section{Feasibility in the Physical Model}
\label{s:feas}

In this section, we derive the graphs $\ghi$ that achieve the sandwiching property for the physical model. The discussion is split into two parts: the case of general thresholds, and the special case of uniform thresholds.
The threshold $\beta_i$ of link $i$ (formally defined below) is the factor by which the interference must be smaller than the signal, in order to have successful transmission in $i$, and 'uniform' refers to the assumption that all links have equal thresholds.

\begin{theorem} The graph $\ghi=G_{\gamma x^\delta}$, for appropriate constants $\gamma>1$ and $\delta\in (0,1)$, together with $\glo$ defined in Sec.~\ref{s:conflict}, gives an $O(\log\log\Ds)$-sandwich of the physical model with general thresholds.
\end{theorem}
As a bonus, this approach uses only \emph{oblivious power assignments}, namely assignments where the power of every link depends only on its length (and global parameters, such as maximum link length). Moreover, the selected graph $\ghi$ guarantees \emph{bi-directional} feasibility, meaning that the feasible sets remain so even after one reverses the directions of a subset of links.

For uniform thresholds, we consider the graph $\ghi=G_{\gamma\tlog}$, for a constant $\gamma>1$, where $\tlog(x)=\max(\log^{t}(x), 1)$, for a constant $t>1$ (specified in Sec.~\ref{s:unithresholds}).
\begin{theorem}
The graph $G_{\gamma\tlog}$, for an appropriate constant $\gamma>1$, together with $\glo$ defined in Sec.~\ref{s:conflict}, gives an $O(\log^*\Ds)$-sandwich of the physical model with uniform thresholds.
\end{theorem}

In this case too, independent sets in $\ghi$ are bidirectionally feasible (potentially using different power assignments for different orientations of links).

Note that by the results of Sec.~\ref{s:conflict}, the tightness provided by the graph $G_{\gamma x^\delta}$ is $O(\log\log\Ds)$, while the tightness provided by $G_{\gamma \tlog}$ is $O(\log^*\Delta)$.
Therefore, the main goal towards the proof of the theorems above is to choose the parameters of the proposed conflict graphs, so as to guarantee feasibility. Thms.~\ref{T:obliviouspowers} and~\ref{T:globalmain} provide these results. Before going on to proofs, we give the formal definitions of the physical model.

\subsection{Model}
\label{s:model}

\textbf{Feasibility.}
The nodes have adjustable transmission power levels.
A \emph{power assignment} for the set $L$ is a function $P:L\rightarrow \mathbb{R}_+$. For each link $i$, $P(i)$\label{G:power} defines the power level used by the sender node $s_i$.
In the \emph{physical model} of communication, when using a power assignment $P$, a transmission of a link $i$ is successful if and only if
\begin{equation}\label{E:sinr}
SIR(S,i)=\frac{P(i)/l_i^{\alpha}}{\sum_{j\in S\setminus \{i\}} P(j)/d_{ji}^{\alpha}}> \beta_i,
\end{equation}
where  $\beta_i\ge 1$\label{G:beta} denotes the minimum signal to noise ratio required for link $i$, $\alpha$\label{G:alpha} is the so-called path loss exponent and $S$ is the set of links transmitting concurrently with link $i$. Here, $P(i)/d^\alpha$ is the power of the sender node $s_i$ received at a distance $d$ from it; hence, the left-hand side of (\ref{E:sinr}) is in fact the ratio of the intended signal over the accumulated interference at the receiver $r_i$.

Note that we omit the noise term in the formula above, since we focus on interference-limited networks. This can be justified by the fact that one can simply slightly decrease the data rates to make the effect of the noise negligible, then restore the rates by paying only constant factors in  approximation. This is further elaborated in the last paragraph below.

A set $S$ of links is called $P$-\emph{feasible} if the condition~(\ref{E:sinr}) holds for each link $i\in S$ when using power assignment $P$. We say $S$ is \emph{feasible} if there exists a power assignment $P$ for which $S$ is $P$-feasible. 

We assume that $\alpha>m$, where $m$ is the doubling dimension of the space; this corresponds to the standard assumption $\alpha >2$ in the Euclidean plane which is necessary to ensure a degree of locality to the communications.

\textbf{Sensitivity.}
We define the sensitivity of link $i$ as $\s_i=\beta_i^{1/\alpha} l_i$. 
We call a set $S$ of links \emph{equilength} if for every two links $i,j\in S$, $\s_i \le 2\s_j$, i.e., $\Ds(S) \le 2$.
Note that with the introduction of sensitivity, the feasibility constraint~(\ref{E:sinr}) becomes: 
\[ \frac{P(i)}{\s_i^\alpha} \ge \sum_{j\in S\setminus \{i\}}\frac{P(j)}{d_{ji}^{\alpha}}\ . \]

\textbf{Power Control.}
Different power control regimes give different notions of feasibility. Guaranteeing just feasibility might require \emph{global power control}, i.e., optimizing the power assignment based on the whole network state. Another option is \emph{oblivious power assignments}, where the power level of a link depends only on local parameters. We will work with a family of oblivious power assignments $P_\tau$ parameterized with $\tau\in (0,1)$, where $P_\tau(i)\sim \s_i^{\tau\alpha}$ for each link $i$. Oblivious power assignments are preferable because they are simple and robust to link churn, but they may give worse performance than global power control.

\textbf{Bi-directional Feasibility.}
Our positive results  hold with a stronger notion of \emph{bi-directional feasibility}, where a set $S$ of links is called bi-directionally feasible if it is feasible and remains so even if we reverse the directions of a subset of links in $S$ (i.e., switch the roles of senders and receivers). 
Note that this might require different power assignments for different orientations. For an oblivious power assignment $P_\tau$, we say a set $S$ is bi-directionally $P_\tau$-feasible if it is $P_\tau$-feasible with every orientation of links.
Note also that bi-directional feasibility is merely a ``bonus'' from our approach; we still compete with the optima of the original model.

\textbf{Ambient Noise and Sensitivity.}
The complete condition for signal reception in the physical model is as follows: $P(i)/\s_i^{\alpha} > \sum_{j\in S\setminus \{i\}} P(j)/d_{ji}^{\alpha} + N_i$, where $N_i\ge 0$ is the \emph{ambient noise} at the receiver $r_i$. We assume, as done in the majority of the related work, that $P(i)\ge cN_i\s_i^\alpha$, holds for a constant $c>1$, and every link $i$. The rationale behind this assumption is to exclude \emph{weak} or \emph{noise-limited} links from consideration, namely links that can tolerate very little interference. Note that $P(i)\ge N_i\s_i^\alpha$ is necessary for the link to be usable even without interference. Scheduling weak links can be considered a separate problem, and is further discussed in Sec.~\ref{ss:weaklb}. 

On the other hand, the assumption $P(i)\ge cN_i\s_i^\alpha$ can be used to suppress the noise term, as follows. Given power assignment $P$ and a number $t>1$, let us call a set $S$ of links \emph{$t$-strong} if $P(i)/\s_i^{\alpha} > t\cdot \sum_{j\in S\setminus \{i\}} P(j)/d_{ji}^{\alpha}+N_i$. Note that a $1/(1-1/c)$-strong set with zero noise ($N_i=0$) is feasible with non-zero noise: $$P(i)/\s_i^{\alpha} -N_i\ge (1-1/c)P(i)/\s_i^{\alpha}>  \sum_{j\in S\setminus \{i\}} P(j)/d_{ji}^{\alpha}.$$
Hence, instead of working with non-zero noise term, we can work with $N_i=0$ and $\s_i'=\s_i\cdot 1/(1-1/c)^{1/\alpha}$. The following result (reformulation of \cite[Cor. 2]{HB15}) shows that this only affects the constant factors in our approximations. It also justifies our assumption in Sec.~\ref{s:conflict} that $\s_i\ge 4l_i$: if the latter does not hold, we simply scale $\s_i$ by a factor of $4$, for all $i$.
 \begin{theorem}\cite{HB15}\label{T:signalstrengthening}
Every $1$-strong set can be partitioned into $\left\lceil 2t\right\rceil$ sets that are $t$-strong.
 \end{theorem}

\subsection{Feasibility of Independent Sets: General Thresholds}

The main technical task is showing feasibility: that each independent set $S$ in $G_f$ corresponds to a feasible set of links. We break the task into bounding the interference of a given link $i$ in $S$ from links in $S$ with more (less) 
sensitivity than $i$ in Lemma~\ref{P:mainlemma1} (Lemma~\ref{P:mainlemma2}), respectively. Both of those lemmas are based on splitting $S$ into sets of roughly equal lengths and bounding the resulting interference as a geometric sum. The bound for the interference from an equilength set is given in Lemma~\ref{P:oblcore}, which itself is a geometric sum, here in terms of the contributions of links of different distances from the link $i$. The actual bound of that inner geometric sum is given in Lemma \ref{L:summation}, which is a variation of a frequently given argument in terms of concentric annuli around the link $i$.
Additionally, we necessarily bound in Lemma \ref{L:distreduction} from below the minimum distance of links that are non-adjacent in $G_f$.
The argument for uniform thresholds (Sec.~\ref{s:unithresholds}) follows the same pattern, but is somewhat simpler.

Now to the formal arguments.
Our goal is to identify constants $\gamma,\delta$ such that
each independent set in $\ghi=G_{\gamma x^\delta}$ is feasible.
We show that this can be achieved by using an oblivious power assignment $P_\tau$, for an appropriate $\tau \in (0,1)$. Moreover, this even holds for bidirectional feasibility.

\begin{theorem}\label{T:obliviouspowers}
Let $\delta_0=\frac{\alpha-m+1}{2(\alpha-m) + 1}$. If $\delta\in (\delta_0,1)$ and the constant $\gamma>1$ is large enough, there is a value $\tau \in (0,1)$ such that each independent set in $G_{\gamma x^\delta}$ is bidirectionally $P_{\tau}$-feasible.
\end{theorem}

First, we introduce a measure of interference under power $P_\tau$.
The interference of link $j$ on link $i$ is 
\[
I_{\tau}(j,i)= \frac{\s_j^{\tau\alpha}\s_i^{(1-\tau)\alpha}}{d_{ji}^\alpha}
 < 1,
\]
when $j\ne i$ and $I_{\tau}(i,i)=0$. 
The interference of the other links in a set $S$ on $i$ is
$I_{\tau}(S,i) = \sum_{j\in S\setminus \{i\}} I(j,i)$.
Showing feasibility of set $S$ under $P_\tau$ is equivalent to showing that
$I_{\tau}(S,i) \le 1$.

We bound this in Lemma~\ref{P:mainlemma1} (Lemma \ref{P:mainlemma2}) for links with less (more) sensitivity than $i$, respectively, and derive the choices for $\tau$ based on those bounds. At high level, the argument proceeds by splitting $S$ into groups of roughly equal length and equal distance from $i$ and bounding the size of these groups as well as their interference on $i$. This is then combined into a double geometric sum that converges to a constant that can be made smaller than one.

\begin{proof}[Proof of Thm.~\ref{T:obliviouspowers}]
Consider any $\delta\in (\delta_0,1)$. Lemmas~\ref{P:mainlemma1} and \ref{P:mainlemma2} bound the interference ($I_\tau$) from links with less and more (resp.) sensitivity than $i$, that form an independent set of links in $G_{\gamma x^\delta}$. 
 If $\gamma$ is sufficiently large, the two bounds add up to less than one. The lemmas require different constraints on $\tau$, but
it can be checked that when $\delta\in (\delta_0, 1)$,  $b:=1-\delta\cdot \frac{\alpha-m}{\alpha} < e:=1-(1-\delta)\cdot\frac{\alpha - m + 1}{\alpha}$, and hence $\tau$ can be chosen to be any point in the interval $(b,e)$.
\end{proof}

 \begin{lemma}\label{P:mainlemma1}
Let $\gamma>2$ and $ \tau > 1- \delta (\alpha-m)/\alpha$. If  $S$ is a set of links that is independent in $G_{\gamma x^\delta}$, and $i$ is a link in $S$ satisfying $\s_i\ge\max_{j\in S}\s_j$, then
 $
 I_{\tau} (S, i)= O\left(\gamma^{-\alpha/2}\right).
 $
 \end{lemma}
\begin{proof}
We partition $S$ into equilength subsets 
$
L_t=\{j\in S: 2^{t-1}\s_0 \leq \s_j<2^t \s_0\},
$ $t=1,2,\ldots,$
 where $\s_0=\min_{j\in S}\{\s_j\}$, and bound each $L_t$ separately. 
The independence condition between $i$ and any other link $j\in S$ is $d_{ij}d_{ji} > \gamma \s_i^{1+\delta}\s_j^{1-\delta}$. Using Lemma~\ref{L:distreduction}, we obtain the more convenient bound $d(i,j)>\gamma' \s_i^{\delta}\s_j^{1-\delta}$, where $\gamma'=\frac{\gamma}{\sqrt{\gamma+1}+1}-1=\Theta(\sqrt{\gamma})$.
Fix a given set $L_t$ and let $\e_t=\min_{j\in L_t}\s_j$. 
We similarly observe that $d(j,k) > \gamma'\e_t$ for all $j,k\in L_t$. Hence, Lemma~\ref{P:oblcore} applies for $L_t$  with parameters $\gamma=\gamma'$ and $\mu=\delta$, giving us the bound
\[
{I_{\tau}(L_t,i)} = O\left(\gamma^{-\alpha/2}\left(\frac{\e_t}{\s_i}\right)^{\mu(\alpha-m) - (1-\tau)\alpha }\right).
\]
We combine the bounds for $L_t$ into a geometric series for $S$:
\[
{I_{\tau}(S,i)} = \sum_{t=1}^{\infty}{I_{\tau}(L_t,i)} 
\le \frac{O(\gamma^{-\alpha/2})}{\s_i^{\mu(\alpha-m)  - (1-\tau)\alpha}}\sum_{t=0}^{\lceil\log{\s_i/\s_0}\rceil}{(2^{t}\s_0)^{\mu(\alpha-m)  - (1-\tau)\alpha}}.
\]
Recall that we assumed $\tau > 1- \delta (1-m/\alpha)$; hence, $\mu(\alpha-m)  - (1-\tau)\alpha> 0$. Thus,
the last sum is bounded by $O(\s_i^{(1-\tau)\alpha - \mu(\alpha-m)})$, which implies the claim.
\end{proof}

 \begin{lemma}\label{P:mainlemma2}
Let $\gamma>2$ and $\tau < 1- (1-\delta)(\alpha - m + 1)/\alpha$. If  $S$ is a set of links that is independent in $G_{\gamma x^\delta}$, and  $i$ is a link in $S$ satisfying $\s_i = \min_{j\in S}\s_j$, then
 $
 I_{\tau} (S, i)= O\left(\gamma^{-\alpha/2}\right).
 $
 \end{lemma}
 \begin{proof}
We proceed as in Lemma~\ref{P:mainlemma1}. Let us split $S$ into equilength subsets $L_1, L_2,\dots$, where
$
L_t=\{j\in S: 2^{t-1}\s_i \leq \s_j<2^t \s_i\}.
$
Let $\e_t=\min_{j\in L_t}\s_j \ge  2^{t-1}\s_i$. Using independence and applying Lemma~\ref{L:distreduction}, we have $d(i,j) > \gamma'\s_j^{\delta}\s_i^{1-\delta}$ for each $j\in S$ (recall that $\s_i\le \s_j$), and $d(j,k)>\gamma' \e_t$ for all $j,k\in L_t$, where $\gamma'=\frac{\gamma}{\sqrt{\gamma+1}+1}-1=\Theta(\sqrt{\gamma})$. We apply Lemma~\ref{P:oblcore} with $\gamma = \gamma'$ and $\mu=1-\delta$ to the set $L_t$ and link $i$ to obtain:
\[ {I_{\tau}(L_t,i)} = O\left(\gamma^{-\alpha/2}\left(\frac{\s_i}{\e_t}\right)^{(1-\tau)\alpha - \mu(\alpha-m+1)}\right) =O\left(\gamma^{-\alpha/2}\left(\frac{1}{2^{t-1}}\right)^{(1-\tau)\alpha - \mu(\alpha-m+1)}\right). \]
Our assumption on $\tau$ implies that $\eta :=(1-\tau)\alpha - \mu(\alpha-m+1) > 0$.
Thus, we have:
$
{I_{\tau}(L,i)}= \sum_{1}^{\infty}{{I_{\tau}(L_t,i)}}
  = O\left(\gamma^{-\alpha/2}\right)\cdot \sum_{t=0}^{\infty}{\frac{1}{2^{\eta t}}} 
  = O\left(\gamma^{-\alpha/2}\right).
$
\end{proof}

The next lemma shows that when two links are independent in the conflict graph $G_f$, 
they must also be well separated in space.

\begin{lemma}\label{L:distreduction}
Let $f$ be a non-decreasing function, such that $f(x)\le \gamma x$, for a constant $\gamma>2$ and for all $x\ge 1$. Let links $i,j$ be independent in $G_f$, and such that $\s_i\ge \s_j$. Then $d(i,j)>\frac{1}{\sqrt{\gamma+1}+1}\s_j f(\s_i/\s_j)-l_j$.
\end{lemma}
\begin{proof}
 Let $D=\max\{d_{ij},d_{ij}\}$ and $d=\min\{d_{ij},d_{ji}\}$. Let $z>2$ be a parameter. Consider the following two cases:
\begin{enumerate}
 \item $D > z\s_i$. The triangle inequality and the assumption $\beta_i,\beta_j\ge 1$ imply that 
\[d(i,j)\ge D - l_i - l_j > (z-2)\s_i \ge \frac{z-2}{\gamma}\cdot \s_j \cdot \gamma(\s_i/\s_j)\ge \frac{z-2}{\gamma}\cdot \s_j f(\s_i/\s_j).\]
\item $D\le z\s_i$. The independence condition implies that $d_{ij}d_{ji}>\s_i\s_j f(\s_i/\s_j)$. Hence, in this case, $d>\frac{1}{z}\cdot \s_j f(\s_i/\s_j)$. By the triangle inequality, $d(i,j)\ge d-l_j$.
\end{enumerate}
Choosing $z=\sqrt{\gamma+1}+1$ implies the claim.
\end{proof}

The next lemma is the common part of Lemmas~\ref{P:mainlemma1} and~\ref{P:mainlemma2}: It bounds the interference from a group of equilength links that are both well separated internally and sufficiently far from the link $i$, and covers both the cases of links more and less sensitive than link $i$. We separate its core technical part into Lemma~\ref{L:summation}, which will be later reused in Sec.~\ref{s:unithresholds}.

\begin{lemma}\label{P:oblcore}
Let $\mu, \tau \in (0,1)$ and $\gamma \ge 1$ be parameters, let $S$ be a set of equilength links such that for all $j,k\in S$, $d(j,k) > \gamma \s_0$, where $\s_0=\min_{j\in S}\s_j$, and let $i$ be a link in $S$ satisfying $d(i,j) > \gamma \s_i^\mu\s_j^{1-\mu}$ for all $j\in S$.
Then,
 \[
 \displaystyle I_{\tau}(S,i)=
 O\left(\gamma^{-\alpha} \left(\frac{\s_i}{\s_0}\right)^{(1-\tau)\alpha - \mu(\alpha-m)}\cdot \min\left\{1,\frac{\s_i}{\s_0}\right\}^{-\mu} \right).
 \]
\end{lemma}
\begin{proof}
Consider first the subset $S' \subseteq S$ of links that are closer to $r_i$ than to $s_i$,
\[
S'=\{j\in S: \min\{d(s_j,r_i), d(r_j,r_i)\} \leq \min\{d(s_j,s_i), d(r_j,s_i)\}\}.
\]
 For each link $j\in S'$, let $p_j$ denote the  endpoint of $j$ that is closest to $i$'s receiver, $r_i$. Denote $q= (\s_i/\s_0)^{\mu}$.
Consider the nested subsets $S_1,S_2,\dots$ of $S'$, where
 \[
 S_r=\{j\in S': d(j,i)=d(p_j,r_i)\leq \gamma(q \s_0+(r-1)\s_0)\}.
 \]
 Note that $S_1=\emptyset$: for every $j\in S'$, $d(j,i)>\gamma q\s_0$. 

Fix $r>1$.
For every $j,k\in S_r$, we have that
 $d(p_j,p_k) \ge d(j,k) > \gamma\s_0$  and that $d(p_j,r_i)\leq \gamma (q \s_0+(r-1)\s_0)$ for each $j\in S_r$ (by the definition of $S_r$). By the doubling property of the metric space, 
\begin{equation}
|S_r|=|\{p_j\}_{j\in S_r}|\leq C\cdot \left(\frac{\gamma(q\s_0+(r-1)\s_0 )}{\gamma\s_0}\right)^{m} = C \left(q+r-1\right)^{m}.\label{E:strs}
\end{equation}
Note also that $\s_j \leq 2\s_0$ and $d(i,j) \ge \gamma(q\s_0+(r-2)\s_0)$ for every link $j\in S_r\setminus S_{r-1}$ with $r>1$; hence,
by the definition of $I_\tau$,
\begin{equation}
I_{\tau} (j, i) \le \frac{\s_j^{\tau\alpha} \s_i^{(1 - \tau)\alpha}}{d(i,j)^\alpha}
 \leq \left(\frac{\s_i}{\s_0}\right)^{(1-\tau)\alpha}\left(\frac{2\s_0}{\gamma(q\s_0+(r-2)\s_0)}\right)^\alpha =\frac{Z_i}{\left(q+r-2\right)^\alpha},\label{E:feqs}
\end{equation}
where $Z_i=(2/\gamma)^\alpha (\s_i/\s_0)^{(1-\tau)\alpha}$. Applying Lemma~\ref{L:summation} to the set $S'=\cup_{r\ge 2}S_r$, with the parameters $q,h=C$ and function $A=I_\tau(\cdot,i)/Z_i$, we get that ${I_{\tau}(S',i)} = O(Z_i)\cdot q^{m-\alpha}\cdot O(\min(1,q)^{-1})$.
which implies the desired bound for the set $S'$ by plugging the values of $q$ and $Z_i$.

The proof holds symmetrically for the set $S \setminus S'$ of links closer to the sender $s_i$ than to the receiver $r_i$. We can define the set $\{p_j\}_{j\in S\setminus S'}$ where $p_j$ is the endpoint of link $j$ that is closest to $r_i$, for each $j\in S\setminus S'$. The rest of the proof will be identical, by replacing $r_i$ with $s_i$ in the formulas.
\end{proof}

The following lemma is used to combine the interference contributions of groups of different distances from $i$ (after combining links of different lengths but within same distance of $i$), represented by the sets $S_r$. Intuitively, if the number of links within a given distance grows polynomially slower than the interference contribution of those links, then the total interference will converge as a geometric sum.

\begin{lemma}\label{L:summation}
Let $q,h>0$ be parameters.
Consider sets $\emptyset=S_1\subseteq S_2\subseteq S_3\subseteq\dots$ of links and let $S = \cup_{r\ge 1} S_r$. 
Suppose $|S_r| \le h\cdot (q+r-1)^m$ and 
assume a function $A:S\rightarrow \mathbb{R}_+$ such that $A(j)\le 1/(q+r-2)^\alpha$,
for all $r>1$ and $j\in S_r\setminus S_{r-1}$. 
Then, $\sum_{j\in S}A(j) = O(h)\cdot q^{m-\alpha}\cdot \min(1,q)^{-1}$.
\end{lemma}
\begin{proof}
First, using the assumptions and sum rearrangements we have (explanations below): 
\begin{align*}
\sum_{S}{A(j)} & = \sum_{S_1}A(j) + \sum_{r\geq 2}{\sum_{S_r\setminus S_{r-1}}{A(j)}} \nonumber \\
& \le \sum_{r\geq 2}\frac{|S_r\setminus S_{r-1}|}{(q+r-2)^\alpha} \nonumber \\
& = \sum_{r\geq 2} \frac{|S_r|- |S_{r-1}|}{(q+r-2)^\alpha} \nonumber \\
& = \sum_{r\geq 2}|S_r| \left(\frac{1}{(q+r-2)^\alpha} - \frac{1}{(q+r-1)^\alpha}\right) \nonumber,
\end{align*}
where the first inequality uses $S_1=\emptyset$ and the upper bounds on $A$, the following equality holds because $\{S_r\}$ is a nested sequence, and the last one is a result of a sum rearrangement and the assumption that $S_1=\emptyset$. 

The convexity of the function $f(x)=x^{-\alpha}$ implies\footnote{For every convex differentiable function $f$ and $x,y\in \textbf{dom}\ f$, $f(x)-f(y)\ge f'(y)(x-y)$.} that $1/(x-1)^\alpha - 1/x^\alpha \le \alpha/(x-1)^{\alpha+1}$. 
Thus, continuing and using the bound on $|S_r|$, 
\[ \sum_{S}{A(j)} \le \alpha \sum_{r\geq 2} \frac{|S_r|}{(q+r-2)^{\alpha+1}} 
 \le \alpha h\cdot \sum_{r\geq 0}{(q+r)^{m-\alpha-1}} \ . \]
To complete the proof, note that $m-\alpha-1<-1$, so the sum converges and can be bounded by an integral, as follows:
\[
\sum_{r\geq 0}{(q+r)^{m-\alpha-1}} \le q^{m-\alpha-1} + \int_{q}^\infty  x^{m-\alpha-1} \, dx 
= q^{m-\alpha-1} + \frac{q^{m-\alpha}}{\alpha-m} = q^{m-\alpha}\cdot O(\min(1,q)^{-1})\ .
\]
\end{proof}

\subsection{Feasibility of Independent Sets: Uniform Thresholds}\label{s:unithresholds}

Now let us consider the special case when all thresholds $\beta_i$ are fixed at a uniform value $\beta\ge 1$, which, however, may be arbitrary, i.e., it may depend on network size or link lengths. We show that in this case feasibility can be guaranteed with much smaller tightness, namely $O(\log^*\Delta)$. This is achieved by choosing a slow-growing function $f$ in the definition of the conflict graphs, and by using global power control.

For dealing with global power control, we use the convenient sufficient condition for feasibility due to Kesselheim \cite{kesselheimconstantfactor}, which defines a purely geometric constraint on the set of links that implies feasibility.
Consider an additive \emph{interference operator} $I$ defined as follows. For links $i,j$, let $I(i,j)=\frac{\beta l_i^\alpha}{d(i,j)^\alpha}=\frac{\s_i^\alpha}{d(i,j)^\alpha}$ and define $I(i,i)=0$ for simplicity of notation.  The operator $I$ is additively expanded: for a set $S$ of links and a link $i$, let $I(S,i)=\sum_{j\in S}I(j,i)$ and $I(i,S)=\sum_{j\in
  S}I(i,j)$.  We will use the notation $I(L)=\max_{i\in L}{I(\{j\in L: l_j\le l_i\},i)}$, which denotes the maximum influence on any link by shorter links in $L$.

\begin{theorem}\cite{kesselheimconstantfactor}\label{T:kesselheimconstant}
For any set $L$ of links  in a metric space, if $I(L)<\frac{1}{12\cdot 3^{\alpha}}$, then $L$ is feasible.
\end{theorem}

We consider the conflict graph $\ghi=G_f$ with $f(x)=\gamma\tlog(x)=\gamma \max(\log^{2/(\alpha - m)}(x), 1)$ for an appropriate constant $\gamma>0$.
We show that for a large enough constant $\gamma> 0$, independence in $G_{\gamma\tlog}$ implies feasibility. In particular, we show that if a set $S$ is independent then $I(S)= O(\gamma^{-\alpha/2})$.
An appropriate choice of $\gamma$ yields feasibility via Thm.\ \ref{T:kesselheimconstant}.

\begin{theorem}\label{T:globalmain}
There is a constant $\gamma > 2$ such that every independent set in $G_{\gamma\tlog}$ is feasible.
\end{theorem}
The idea behind the proof is similar to the one for general thresholds. The main difference is that the absence of arbitrary (non-geometric) thresholds $\beta_i$ allows us to obtain feasibility with much less independence than before.

Again, in order to bound the interference of an independent set of links on a (longer) link $i$, we split the whole set into equilength subsets, bound the interference of each equilength subset using Lemma~\ref{L:pcequilength}, and combine them into a series that converges when we choose $f(x)=\gamma\tlog(x)$. All this is done in the following lemma.

 \begin{lemma}\label{L:globalmain}
   Let $S$ be an independent set in $G_{\gamma\tlog}$ with $\gamma> 2$. 
Then
 $
 I(S) = O\left(\gamma^{-\alpha/2} \right).
 $
\end{lemma}

  \begin{proof}
Fix an arbitrary link $i\in S$, and denote $S_i^-=\{j\in S: l_j\le l_i\}$. It suffices to show that $I(S_i^-,i)=O\left(\gamma^{-\alpha/2} \right)$. 
	
  Let $\gamma'$ be such that $\gamma\tlog(x) \le \gamma' x$, for every $x\ge 1$. Consider any link $j\in S_i^-$. Using independence and applying Lemma~\ref{L:distreduction}, we obtain that $d(i,j) > \gamma'' \s_{j}\tlog(l_i/l_j)$, where $\gamma''=\frac{\gamma}{\sqrt{\gamma'+1}+1}-1=\Theta(\sqrt{\gamma})$. Similarly, for every pair of links $j,k\in S_i^-$ with $l_j\ge l_k$, we have $d(j,k)>\gamma'' \s_{k}\tlog(l_i/l_j)\ge \gamma''\tlog(1) \cdot \s_{k}$.

Partition $S_i^-$ into equilength subsets 
$S_t = \{j\in S_i^- : l_i/2^{t-1} \leq l_j < l_i/2^t\},$ $t=1,2,\ldots$
Let $\ell_t$ denote the smallest link length in $S_t$.
The conditions of Lemma~\ref{L:pcequilength} hold for each $S_t$, so applying the lemma gives us the bound:
\[
I(S_t,i) = O(1)\cdot (\gamma''\tlog(l_i/\ell_t))^{m-\alpha}/(\gamma'')^{m}=O(\gamma^{-\alpha/2})\cdot \tlog(l_i/\ell_t)^{m-\alpha}.
\]
Observe that $\log (l_i/\ell_t) = t$, so $\tlog(l_i/\ell_t) = t^{2/(\alpha-m)}$.
Thus, 
\[
I(S_i^-,i)= \sum_{t\ge 1}{I(S_t,i)}\le O(\gamma^{-\alpha/2})\cdot\left(1+ \sum_{t\ge 1}{\left(t^{2/(\alpha-m)}\right)^{m-\alpha}}\right)=O(\gamma^{-\alpha/2}) \sum_{t\ge 1} t^{-2} = 
O(\gamma^{-\alpha/2}) \ , \]
which completes the proof.
 \end{proof}

The following lemma is the analogue of Lemma~\ref{P:oblcore}: It bounds the interference of an equilength independent set $S$ on a longer link $i$ that is independent from the set $S$. 

\begin{lemma}\label{L:pcequilength}
 Let $f$ be a non-decreasing function, such that $f(x)\geq 1$ whenever $x\geq 1$. Let $S$ be an equilength set of links such that for all $j,k\in S$ with $d(j,k)>f(1) \cdot \min\{\s_j,\s_k\} $, and let $i$ be a link in $S$ such that, for each $j\in S$, $l_i\ge l_j$ and $d(i,j)>\s_j f(l_i/l_j)$.
Then
 $
  I(S,i) = O(1)\cdot f(l_i/\ell)^{m - \alpha}/f(1)^m,
 $
where $\ell$ denotes the largest link length in $S$.
 \end{lemma}

\begin{proof}
\def \hel {\hat \ell}
We proceed as in the proof of Lemma~\ref{P:oblcore}.
We treat the subset $S'$ of links in $S$ that are closer to $r_i$ than to $s_i$ and bound $I(S',i)$, while the symmetric case of $S\setminus S'$ is omitted.

Let us denote $q=f(l_i/\ell)$ and $\hel = \beta\ell$. Note that $q\geq 1$ because $l_i/\ell\ge 1$. 
For a link $j\in S'$, let $p_j$ denote the endpoint of link $j$ that is closest to $r_i$, i.e., $d(i,j)=d(p_j,r_i)$. 
Consider the nested sequence of subsets $S_1\subseteq S_2\subseteq \dots\subseteq S'$, where 
 \[
 S_r=\{j\in S': d(j,i)=d(p_j,r_i)\leq q \hel/2+(r-1)\hel/2\}.
 \]
We will apply Lemma~\ref{L:summation}.
First, note that $S_1=\emptyset$:  $S'$ is an equilength set with maximum link length $\ell$ and $f$ is non-decreasing, so by our independence assumption, $d(i,j)>\s_j f(l_i/l_j)\ge q \hel/2$.

Next, let us bound $|S_r|$ using the doubling dimension of the metric space. Consider any $j,k\in S_r$ such that $l_j\geq l_k$. By our assumption, $d(p_j,p_k) \ge d(j,k)> f(1) \cdot \min\{\s_j,\s_k\}\geq f(1)\hel/2.$
 By the definition of $S_r$, $d(p_j,r_i)\leq q \hel/2+(r-1)\hel/2$ for each $j\in S_r$. Thus,
\[
|S_r|=|\{p_j\}_{j\in S_r}|< C\cdot \left(\frac{ q\hel/2+(r-1)\hel/2 }{f(1)\hel/2}\right)^{m} = \frac{C}{f(1)^m} \left(q+r-1\right)^{m}.
\]

Next, let us bound $\max_{j\in S_{r}\setminus S_{r-1}}\{I(j,i)\}$. For each $r>1$ and for any link $j\in S_r\setminus S_{r-1}$, we have that $l_j \leq \ell$ and $d(i,j) > q\hel/2+(r-2)\hel/2$; hence, 
$
I (j, i) = \frac{\s_j^{\alpha}}{d(i,j)^\alpha}
 < \left(\frac{\hel}{q\hel/2+(r-2)\hel/2}\right)^\alpha
 =\frac{2^\alpha}{\left(q+r-2\right)^\alpha}.
$

Thus, we apply Lemma~\ref{L:summation} to the set $S'=\cup_{r\ge 2}S_r$ with the parameters 
$q,h=C/f(1)^m$ and function $A=I(\cdot,i)/2^\alpha$, to obtain:
 $I(S',i) =O(q^{m-\alpha}/f(1)^m)$, where we also use the fact that $q\ge 1$.
\end{proof}

\section{Limitations of the Conflict Graph Method}
\label{s:limitations}

Our results are best possible (up to constant factors) in several different ways.
Specifically,
conflict graphs for the physical model:
\begin{itemize}[noitemsep]
 \item are necessarily of the form we consider (for uniform data rates) (Sec.~\ref{ss:necessity}),
 \item incur cost of $\Omega(\log\log \Delta)$-factor for general data rates (Sec.~\ref{ss:genlb}),
 \item incur cost of  $\Omega(\log^* \Delta)$-factor for uniform data rates (Sec.~\ref{ss:uniflb}),
 \item give only the trivial approximation (of $n$) as a function of the number of links (Sec.~\ref{ss:uniflb}),
 \item give no approximation (in terms of $\Delta)$ in some non-doubling metrics (Sec.~\ref{ss:genmetrics-lb}),
 \item cannot represent uniform power (with non-trivial tightness) (Sec.~\ref{ss:uniformlb}), and
 \item cannot represent noise-limited networks (Sec.~\ref{ss:weaklb}).
\end{itemize}

Note that the instances we construct are embedded on the real line, i.e., in one dimensional space. 

\subsection{Necessity}
\label{ss:necessity}

What kind of graphs are conflict graphs? By a ``conflict graph formulation'' we mean a
deterministic rule for forming graphs on top of a set of links.  For it to be meaningful as a general purpose
mechanism, such a formulation cannot be too context-sensitive.  We shall postulate some axioms
(that by nature should be self-evident) that lead to a compact description of the space of
possible conflict graph formulations. For simplicity, we focus on the case of uniform thresholds, also because
this case lies at the heart of all our limitation results.

\begin{axiom}
A conflict graph formulation is defined in terms of the \emph{pairwise relationship} of links.
\label{axiom:pairwise}
\end{axiom}

By nature, graphs represent pairwise relationships; conflict graph formulations are boolean predicates of
pairs of links.  More specifically, though, we expect the conflict graph to be defined in terms of the
relative standings of the link pairs. That is, the existence of an edge between link $i$ and link $j$ should
depend only on the properties of the two links, not on other links in the instance.  The only properties
of note are the $\binom{4}{2} = 6$ distances between the nodes in the links.

We refer to a \emph{conflict} between two links if the formulation specifies them to be adjacent in the
conflict graph; otherwise, they are \emph{conflict-free}.

\begin{axiom}
A conflict graph formulation is \emph{invariant to translation and scaling}. 
Translating or scaling links by a fixed factor does not change the conflict relationship.
\label{axiom:scale-free}
\end{axiom}

An essential feature of the SINR formula -- that distinguishes it from
other formulations, like unit-disc graphs -- is that only relative
distances matter.  Thus, the positions of the nodes should not matter,
only the pairwise distances, and only the relative factors among the
distances.  There is a practical limit to which links can truly grow,
due to the ambient noise term.  However, that only matters when
lengths are very close to that limit; we will treat that case
separately.

\begin{axiom}
A conflict formulation is \emph{monotonic} with increasing distances.
\label{axiom:monotonicity}
\end{axiom}

The reasoning is that a conflict formulation should represent the degree of conflict between pairs of
links, or their relative ``nearness''.  Specifically, if two links conflict and their separation (i.e.,
one of the distances between endpoints on distinct links) decreases while the links stay of the same
length, then the links still conflict.  Similarly, if two links are conflict-free and the length of one
of them decreases (while their separation stays unchanged), the links stay conflict-free.

\begin{axiom}
A conflict formulation should respect pairwise incompatibility. That is, 
if two links cannot coexist in a feasible solution, they should be adjacent in the conflict graph.
\label{axiom:incompatible-pair}
\end{axiom}

\smallskip

In the case of conflict graphs for links in the SINR model with arbitrary power control, we propose an
additional axiom.

\begin{axiom}
  A conflict formulation for links under arbitrary power control is
  symmetric with respect to senders and receivers.
\label{axiom:symmetry}
\end{axiom}

Namely, it should not matter which endpoint of a link is the sender and which is the receiver when
determining conflicts. The key rationale for this comes from Kesselheim's sufficient condition for feasibility,
given here as Thm.~\ref{T:kesselheimconstant}. 
As we show in Sec.\ \ref{ss:uniflb}, this formula is also a necessary condition in doubling metrics, up to constant factors.
Thus, feasibility is fully characterized by a symmetric rule (modulo constant factors).

As we shall see, the axioms and the properties of doubling metrics imply that only two distances really
matter in the formulation of conflict graphs: the length of the longer link, and the distance between the
nearest nodes on the two links (both scaled by the length of the shorter link). 

We now argue that all conflict formulations satisfying the above axioms are essentially of the form
$G_f$, for a function $f$ (as defined in (\ref{eq:unifgraphdef})). They can only differ from $G_f$ by what can be accounted for by an appropriate constant factor in the definition of $f$. 

\begin{proposition}
Every conflict graph formulation $\cK$ is captured by $G_f$, for some non-decreasing 
function $f$.
Namely, $\cK$ is sandwiched by  $G_f$ and $G_{g}$, i.e.,
$G_f(L) \subseteq \cK(L) \subseteq G_{\gamma f}(L)$, for every link set $L$, where $g(x)=c_1f(c_2x)$, for some constants $c_1,c_2>0$.
\label{prop:gf-suffices}
\end{proposition}

\begin{proof}
By Axiom \ref{axiom:pairwise}, $\cK$ is a function of link pairs, more specifically, the
distances among the four points. By Axiom \ref{axiom:scale-free}, we can use normalized distances, and will 
choose to factor out the length of the shorter link.
By Axiom \ref{axiom:symmetry}, it does not matter which of them involve senders and which involve receivers.

Now, consider two links $i = (s_i,r_i)$ and $j = (s_j,r_j)$, where $l_i \le l_j$.
Let us denote for short $d = d(i,j)$.
We aim to show that decisions regarding adjacency in $\cK$
can be determined in terms of constant multiples of $d$ and $l_j$.

First, recall that by Axiom \ref{axiom:incompatible-pair}, pairwise incompatible links must be adjacent in any conflict
graph. As observed in Sec.~\ref{s:conflict}, this is encoded in $\glo$, so we may restrict attention to independent sets in $\glo$,
where we have $d_{ij}d_{ji} \ge \s_i\s_j\ge 16l_il_j$. Let us first show that in this case, $d \ge 3l_i$. Indeed, assuming the opposite, and assuming w.l.o.g.\ $d_{ij} \ge d_{ji}$, we obtain, by triangular inequality, that either $d_{ji}= d <3l_i$ and $d_{ij}\le d + l_i + l_j < 5l_j$, or $d_{ji}\le d+l_i<4l_i$ and $d_{ij}\le d+l_j<4l_j$, both cases contradicting the independence assumption. 

We can relate the other distances to a combination of $d(i,j)$ and $l_j$,
using triangular inequality. Assume w.l.o.g.\ that $d=d(s_i,s_j)$. 
First, observe that the distance $d(r_i,s_j)$ is at most constant times the distance
from $i$ to $j$, i.e., $d(s_i,s_j) \le d(r_i,s_j) \le d + l_i \le (1 + 1/3) d(s_i,s_j)$.
Next, we claim that $d(s_i,r_j)$ and $d(r_i,r_j)$ are within a constant multiple of $q = \max(d, l_j/2)$. We have:
$d(s_i,r_j) \ge l_j - d(s_i,s_j) = l_j - d$, and hence $d(s_i,r_j) \ge \max(d,l_j-d) \ge q$, and also $d(s_i,r_j) \le d + l_j \le 3q$. 
Similarly, $d(r_i,r_j)\ge q$ and $d(r_i,r_j) \le d + l_i + l_j \le 4q$.
It follows that all four distances between endpoints are within constant multiples of $d(i,j)$ and $l_j$.

Hence, by monotonicity (Axiom \ref{axiom:monotonicity}), $\cK$ is dominated by a conflict graph
formulation $\cH$ defined by a monotone boolean predicate of two variables: length of the longer link
$l_j$, and the distance $d(i,j)$ between the links (scaled by the shorter link). 
However, an arbitrary monotone boolean predicate of two variables $x, y$ can be represented by a 
relationship of the form $y > f(x)$, for some monotonic function $f$. 
Thus, $\cK$ is dominated by $G_f$, for some non-decreasing function $f$.
Also, by the same arguments, $\cK$ dominates $G_{c_1f(c_2x)}$ for constants $c_1,c_2>0$.
\end{proof}

Finally, we can observe that sub-linearity is necessary if one seeks non-trivial approximations. 
Namely, linear functions correspond to disc graphs, and Moscibroda and Wattenhofer \cite{moscibrodaconnectivity} gave an instance of a feasible set of links that induces a clique in disc graphs. The length diversity $\Delta$ in their construction is $2^n$, 
thus the best approximation one can hope for is logarithmic in $\Delta$.
This holds equally for any super-linear function, by the same construction.

\subsection{Optimality of Tightness Bounds}

We begin with a general result showing that the tightness bound of $O(f^*(\Delta))$ in Thm.~\ref{T:sandwich} cannot be improved. The result further shows that for the physical model, there is no choice of the lower bound graph $\glo$ that can improve the tightness bound.
Finally, this result shows that the number of links $n$ alone is not a meaningful measure for tightness.
The result holds even for the special case when the links have uniform thresholds, and the nodes are arranged on the real line.

The construction is based on the following technical observations. On the one hand, it follows from Thm.\ \ref{T:kesselheimconstant}
that any set of exponentially growing links arranged sequentially by the order of length on the real line is (almost)
feasible. On the other hand, given such a set $S$ of links on the line, a new link $j$ can be formed so that
$j$ is adjacent (in $G_f$) to all the links in $S$ while the set $S\cup j$ stays feasible; the only requirement is that $j$
be long enough. Our construction then builds recursively on these ideas.
\begin{theorem}\label{T:ndependence}
Let $f(x)=\omega(1)$, and assume uniform threshold $\beta\ge 1$. For any integer $n > 0$, there is a \emph{feasible} set $L$ of $n$ links arranged on the real line, such that $G_f(L)$ is a clique, i.e., $\chi(G_f(L))=n$. Moreover, if $f(x)\ge g(x)$ ($x\ge 1$) for a strongly sub-linear increasing function $g(x)$ with $g(x)=\omega(1)$, then $n = \Omega(g^*(\Delta))$.
\end{theorem}

\begin{proof} Let us assume, for simplicity, that $\beta=1$; the argument extends straightforwardly to any constant $\beta> 1$.
Consider a set of $2n$ links, $\{1,2,\dots,2n\}$, arranged consecutively from left to right on the real line. For each $i=1,2,\dots,n-1$, the node $r_i$ is to the right of $s_i$, and shares the same location  on the line with $s_{i+1}$,  i.e., $r_i=s_i+l_i=s_{i+1}$. See Figure~\ref{fig:notcomplicated}.
 The lengths of the links are defined inductively, as follows. We set $l_1=1$, and for $i\ge 1$, we choose $l_{i+1}$ to be the minimum value satisfying:
\begin{align}
\label{E:feasibility} l_{i+1}&\ge cl_i \\ 
\label{E:clique}2d(i+1,j)=2d_{i+1,j}&\le l_jf(l_{i+1}/l_j)\mbox{ for all } j\leq i,
\end{align}
where $c\ge 2$ is a large enough constant, specified below. Such a value of $l_{i+1}$ can be chosen as follows. By the inductive hypothesis, we have $l_j\ge cl_{j-1}\ge 2l_{j-1}$ for $j=2,3,\dots,i$, which implies that $l_i> \sum_{j=1}^{i-1}{l_j}$. Then, we have that $d_{i+1,j}=\sum_{t=j+1}^i{l_t}< 2l_i$ for all $j\le i$. Thus, it is enough to choose $l_{i+1}$ such that $l_{i+1}\ge cl_i$ and $4l_i \le l_jf(l_{i+1}/l_j)$ for all $j\le i$, which can be done using the assumption $f=\omega(1)$ and the fact that the values of $l_j$ for $j\le i$ are already fixed at this point. This completes the construction.
\begin{figure}[htbp]
\begin{center}
\includegraphics[width=0.8\textwidth]{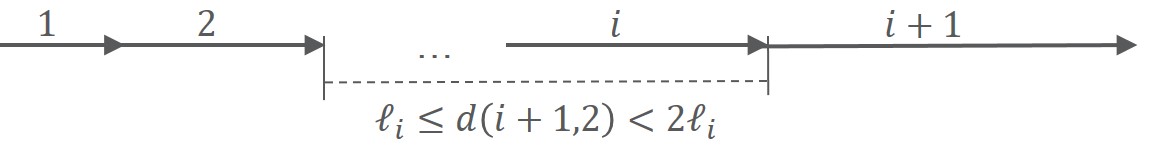}
\caption{The construction in Thm.\ \ref{T:ndependence}.}
\label{fig:notcomplicated}
\end{center}
\end{figure}

First, observe that (\ref{E:clique}) implies that $G_f(L)$ is a clique. Indeed, consider two links $i,j$, such that $i>j$. Then $d_{ji}\le 2l_i$, which, multiplied with (\ref{E:clique}), shows that links $i,j$ are adjacent in $G_f$ (recall that $\beta=1$).

Next, we prove feasibility. Consider the odd numbered links $S=\{1,3,\dots,...,2n-1\}$. Let us fix a link $2k+1\in S$. Let $T=\{j\in S: l_j< l_{2k+1}\}$.
Note that for each $j\in T$, $d(j,2k+1)\ge l_{2k}\ge c^{2k-j-1}$. We have that
\[
I(T,2k+1)=\sum_{j\in T}{\frac{l_j^{\alpha}}{d(j,2k+1)^{\alpha}}}< \sum_{t=1}^\infty {c^{-t\alpha}},
\]
where the last sum is a geometric series that can be made smaller than any constant, by choosing constant $c$ appropriately. Thus, there is a choice of constant $c$ for which $T$ is feasible, as per Thm.\ \ref{T:signalstrengthening}. This proves the first part of the theorem.

Now, let us assume that $f(x)\ge g(x)$ for a strongly sub-linear function $g(x)$ with $g(x)=\omega(1)$. Then, there is a constant $x_0$ such that $g(x) < x$ for all $x \ge x_0$ (because $g(x)=o(x)$) and there is a constant $c'$ such that $2g(x)/x\le g(y)/y$ whenever $x\ge c'y$ (strong sub-linearity). In this case, we repeat the construction above with a few modifications.
We set $l_1=\max(c',x_0)$ and set $l_{i+1}$  be the minimum value s.t. $g(l_{i+1}) \ge cl_i$, for $i=1,2,\dots$ (such a value exists because $g(x)=\omega(1)$), where $c$ is the constant from (\ref{E:feasibility}). Let us show that the conditions (\ref{E:feasibility}-\ref{E:clique}) hold for these links.
  
Since $l_{i+1}\ge x_0$, we have that $l_{i+1} > g(l_{i+1}) \ge cl_i$, which implies (\ref{E:feasibility}). This in turn implies, as observed in the first part of the proof, that $d(i+1,j) < 2l_i$ for all $2\le j\le i$. Let us denote $x=l_{i+1}/l_1=l_{i+1}$ and $y=l_{i+1}/l_j$. Note that $x/y=l_j \ge c'$, so we have, by strong sub-linearity of $g$,  that $g(y)/y \ge 2g(x)/x$, or equivalently, that $l_j\cdot g(l_{i+1}/l_{j}) \ge 2\cdot g(l_{i+1})$; hence $l_j\cdot g(l_{i+1}/l_j) \ge 4l_i > 2d(i+1,j)$ for all $2\le j \le i$, and (\ref{E:clique}) holds. 

It remains to prove the lower bound on $n$. Recall that the value of $l_{i+1}$ is the minimum satisfying $g(l_{i+1})\ge 2l_i$, for $i=1,2,\dots,n-1$. Then, we have $g(l_{i+1}/2) < 2l_i$ or, equivalently, $h(l_{i+1}/2) < l_i/2$, where $h(x)=g(x)/4$. Thus,
\[
1/2=l_1/2 > h(l_2/2) > h(h(l_3/2))> \dots > h^{(n-1)}(l_n/2)=h^{(n-1)}(\Delta/2),
\]
  which implies that $n =\Omega( h^*(\Delta/2))=\Omega(g^*(\Delta))$.
\end{proof}

\subsubsection{Optimality for General Thresholds}
\label{ss:genlb}

Here we show that the obtained tightness is essentially best possible, by demonstrating that every reasonable conflict graph formulation must incur an $O(\log\log\Ds)$ factor.
First, since the  feasibility of a set of links is precisely determined by the values $\s_i$ and $d_{ij}$, we can assume, by a similar reasoning as in Sec.~\ref{ss:necessity}, that the conflict relation is a function of $\frac{\s_{max}}{\s_{min}}, \frac{d_{ij}}{\s_{min}}, \frac{d_{ji}}{\s_{min}}$, where $\s_{min}$ and $\s_{max}$ are the smaller and larger values of $\s_i,\s_j$, respectively. Our construction will consist of only \emph{unit-length} links (i.e.\ $l_i=1$) of mutual distance at least 3. In this case, we can further reduce the number of variables by observing that in such instances, $d_{ij}=\Theta( d_{ji})=\Theta(d(i,j))$. Thus, the conflict relation is essentially determined by two variables: $\frac{d(i,j)}{\s_{min}}$ and $\frac{\s_{max}}{\s_{min}}$. By separating the variables, the conflict predicate boils down to a relation
$
\frac{d(i,j)}{\s_{min}} > f(\frac{\s_{max}}{\s_{min}})
$
for a function $f$. 

Let us show that 
feasibility of independent sets requires that $f(x)=\Omega(\sqrt{x})$ in such a graph.
Let us fix a function $f:[1,\infty)\rightarrow [1,\infty)$. Let $i,j$ be unit-length links with $\beta_j=1$ and $\beta_i=X^\alpha >1$, where $X$ is a parameter. Further assume that the links $i,j$ are placed on the plane so that $d(i,j)=f(X)+3=f(\s_i/\s_j)+3$, which implies that the links are non-adjacent in $G_f$. Thus,  $i,j$ must form a feasible set: $\frac{P(i)}{\s_i^\alpha} > \frac{P(j)}{d_{ji}^\alpha}$ and $\frac{P(j)}{\s_j^\alpha} > \frac{P(i)}{d_{ij}^\alpha}$. Multiplying these inequalities together and canceling $P(i)$ and $P(j)$ out, gives: $d_{ij}d_{ji} > \s_i\s_j=X$. Since the links have unit lengths, while $d(i,j)>2\max(l_i,l_j)$, the triangle inequality implies that $d(i,j)= \Theta(\max(d_{ij},d_{ji}))=\Theta(\sqrt{d_{ij}d_{ji}})=\Omega(\sqrt{X})$, which in turn implies that $f(X)=d(i,j)-3=\Omega(\sqrt{X})$.

Now, the main claim of this section, that is, the tightness must be at least $\Omega(\log\log\Ds)$, follows from Thm.~\ref{T:ndependence}, because for $f(x)=\Omega(\sqrt{x})$, we have  $f^*(x)=\Omega(\log\log x)$.

\subsubsection{Optimality for Uniform Thresholds}
\label{ss:uniflb}

The strategy  of proving a lower bound on $f$ for which $G_f$ is a ``working'' conflict graph, used in the case of general thresholds, seems difficult to apply for the uniform thresholds case. Instead, our strategy here is as follows. First, we observe that for $f(x)=\Omega(\log^{(c)}x)$ with any constant $c\ge 1$, the lower bound $\Theta(\log^*{\Delta})$ on tightness follows from Thm.~\ref{T:ndependence}. In particular, our analysis of $G_{\gamma\tlog}$ is tight. This, however, leaves the possibility that a slower-growing function $f$ could give better tightness. To close this gap, we prove the following theorem.

\begin{theorem}\label{T:hardinstance}
Let $f(x)=O(\log^{1/\alpha}x)$. Assume uniform and fixed thresholds $\beta\ge 1$. For each $\Delta>0$, there is a set $L$ of links on the real line with $\Delta(L)=\Omega(\Delta)$, such that $G_f(L)$ has no edges, but $L$ cannot be partitioned into fewer than $\Theta(\log^*{\Delta(L)})$ feasible subsets.
\end{theorem}
The construction follows the general structure  of a lower bound for scheduling the edges of a minimum spanning tree of a set of points in the plane~\cite[Thm.\ 7]{SODA12}. In order to prove the theorem, we need a necessary condition for feasibility, which we present first. We show in the following result that the sufficient condition for feasibility stated in Thm.\ \ref{T:kesselheimconstant} is essentially necessary in doubling metric spaces. 
This result is of independent interest, as it may prove useful for improved analysis of various problems.
It should be noted that this theorem does not hold in general metric spaces (as opposed to Thm.\ \ref{T:kesselheimconstant}).

\begin{theorem}\label{T:necessary}
Let $\beta\ge 3^\alpha$, and $L$ be a feasible set of links. Then, $I(L) = O(1)$.
\end{theorem}
\begin{proof}
The proof consists of two parts, bounding the interference on a link $i$ by faraway links (i.e., links that are highly independent from link $i$) on one hand, adapting the proof of Lemma~\ref{L:globalmain}, and by near links (the rest) on the other hand, using simple manipulations of the SINR condition.

Let us fix a link $i\in L$ and denote $S=\{j\in L : l_j\le l_i\}$. Let constant $c$ be such that $(c-1)x\ge \beta^{1/\alpha}\tlog(x)$, for each $x\ge 1$.
We split $S$ into two subsets, $S_1=\{j\in S: \max\{d_{ij},d_{ji}\} > cl_i\}$ and $S_2=S\setminus S_1$. We bound the interference on $i$ from $S_1$ and $S_2$ separately.

For $S_1$, we  adapt the proof of Lemma~\ref{L:globalmain}. Feasibility implies independence in $\glo$: For each pair $j,k\in S_1$, $d_{kj}d_{jk} > \beta^{2/\alpha}l_jl_k\ge 9l_jl_k$. We claim that the latter implies that $d(k,j)>2\min(l_k,l_j)=2\tlog(1)\cdot \min(l_k,l_j)$. Assume, w.l.o.g.,\ that $l_j\le l_k$. Assume, for contradiction, that $d(j,k)\le 2l_j$. Then by the triangle inequality, we have $\min(d_{jk},d_{kj})\le l_j + d(j,k)\le 3l_j$, and $\max(d_{jk},d_{kj})\le l_k + d(j,k)\le 3l_k$, contradicting independence. On the other hand, $d(i,j)\ge \max\{d_{ij},d_{ji}\}-l_i> (c-1)l_i\ge \s_j\tlog(\s_i/\s_j)$ holds for each $j\in S_1$, by the definition of $S_1$. We can proceed now as in the proof of Lemma~\ref{L:globalmain}: Partition $S_1$ into equilength subsets, apply Lemma~\ref{L:pcequilength} to each of those (we have just shown that the assumptions of the lemma hold), and combine the obtained bounds into a convergent series. We omit the technical details.

Now, consider the set $S_2$. Let $P$ be a power assignment for which $L$ is $P$-feasible.
By the definition of SINR feasibility,
\[
\frac{P(i)}{l_i^\alpha} > 3^{\alpha}\sum_{j\in S_2}{\frac{P(j)}{d_{ji}^\alpha}},\mbox{ and }\frac{P(j)}{l_j^\alpha} >  3^{\alpha}\frac{P(i)}{d_{ij}^\alpha}\mbox{ for all }j\in S_2.
\]
 By replacing $P(j)$ with $3^{\alpha}\frac{P(i)l_j^{\alpha}}{d_{ij}^\alpha}$ in the first inequality and simplifying the expression, we get:
\begin{equation}\label{E:equation2}
\sum_{j\in S_2}\frac{l_i^\alpha l_j^\alpha}{d_{ij}^\alpha d_{ji}^\alpha} \le 9^{-\alpha}.
\end{equation}
In order to extract a bound on $I(S_2,i)$ from (\ref{E:equation2}), we it suffices to show that for each $j\in S_2$, $\max\{d_{ij},d_{ji}\}=\Theta(l_i)$ and $\min\{d_{ij},d_{ji}\}=\Theta(d(i,j))$. 

Assume, w.l.o.g., that $d_{ij}\ge d_{ji}$. First, as it was observed above, feasibility implies that $d_{ji}\ge d(i,j) > 2l_j$. Hence, $d_{ji}\ge d(i,j)\ge d_{ji} - l_j > d_{ji}/2$. Next, consider $d_{ij}$. 
 Recall that $d_{ij} \le cl_i$, by the definition of $S_1$. To prove $d_{ij}=\Omega(l_i)$, consider two cases: If $l_j\ge l_i/2$, then $d_{ij} \ge d(i,j)> 2l_j\ge l_i$, and otherwise the triangle inequality implies $d_{ij}\ge l_i-d(i,j)>l_i/2$.
\end{proof}

\noindent \emph{Remark.}
Note that Thm.\ \ref{T:signalstrengthening} implies that any feasible set can be refined into a constant number of $3^\alpha$-feasible subsets. Thus, the interference function $I$ fully captures feasibility in doubling metrics, modulo constant factors.

\begin{proof}[Proof of Thm.~\ref{T:hardinstance}]
For a set $S$ of links, we will use $diam(S)$ to denote the \emph{diameter} of $S$, or the maximum distance between nodes in $S$. We assume, for simplicity, that $\beta=3^\alpha$. The argument can easily be extended to any other value $\beta\ge 1$.

We will construct a set of links that cannot be partitioned/scheduled in fewer than $\Theta(\log^*{\Delta})$ feasible slots, relying on the necessary condition for feasibility (Thm.\ \ref{T:necessary}). 

  Let us fix a function $f$. Note that since $f=O(\log^{1/\alpha})$, there is a constant $C\geq 1$ s.t. $f(x)\le C\log^{1/\alpha}x$.  We construct sets  $L_t$ of links, $t=1,2,\dots$, recursively. The construction is illustrated in Figure~\ref{fig:complicated}. All the links are arranged on the real line and the receiver of
  each link is to the right of the sender.  Initially, we have a set $L_1$ consisting of a single
  link of length $1$, for which a single slot is sufficient and necessary.  Suppose that we have already
  constructed $L_{t}$ with the property that at least $t$ slots are required for scheduling $L_{t}$. The
  instance $L_{t+1}$ is constructed as follows, using $k$ scaled copies of $L_t$, where $k$ is to be
  determined. First, we place a single very long link $j_{t+1}$ in the line. We then add, in order from left to
  right, copies 
  $L_t^1,L_t^2,\dots,L_t^{k}$ of $L_t$ to the right of $j_{t+1}$, where $L_t^s$ is the copy of $L_t$ scaled by a factor
  $8^s$. The aim is to ensure the following properties: 
  \begin{enumerate}[noitemsep,label=\roman*]
  \setlength{\parskip}{0cm}%
  \item   {$L_{t+1}$ is $f$-independent,}\label{EN:independence}
  \item {$t=\Omega(\log^*{\Delta(L_t)})$,} \label{EN:lowerbound}
  \item {for any set $S=\{i_1,i_2,\dots,i_k\}$
  with $i_s\in L_t^s$, $s=1,2,\dots,k$, we have that $I(S,j_{t+1})>c_0$, for a constant $c_0$ of our choice.}\label{EN:inconsistency2}
 \end{enumerate}
  The last property ensures that each partitioning of $L_{t+1}$ into feasible subsets must put a complete copy $L_t^s$ in a slot separate from $j_{t+1}$. Indeed, the existence of such a partitioning that placed at least one link from each copy $L_t^s$ in the same slot with $j_{t+1}$ would contradict (\ref{EN:inconsistency2}): we would have $I(S,j_{t+1}) =O(1)$ for some $S$ as above, due to Thm.\ \ref{T:necessary}. Recall that $L_t$ needs at least $t$ slots to be scheduled, and so does each copy of it. It follows that $L_{t+1}$ needs at least $t+1=\Omega(\log^*{\Delta(L_t)})$ slots to be scheduled, one for $j_t$ and at least $t$ for scheduling the copies of $L_t$. Proving the properties (\ref{EN:independence}-\ref{EN:inconsistency2})  will complete the proof of the theorem.
\begin{figure}[htbp]
\begin{center}
\includegraphics[width=0.85\textwidth]{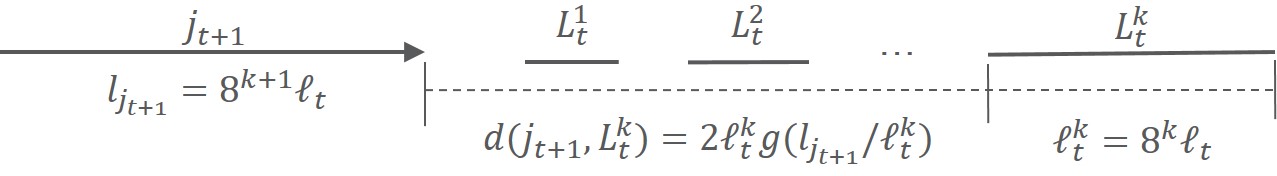}
\caption{The recursive construction of $L_{t+1}$.}
\label{fig:complicated}
\end{center}
\end{figure}

Now let us describe the inductive step of the construction in detail. Let $\ell_t=diam(L_t)$ denote the diameter of $L_t$. The number of copies of $L_t$ is $k=2^{c\ell_t}$, for a large enough constant $c$. The length of link $j_{t+1}$ is set to $l_{j_{t+1}}=8^{k+1}\ell_t$. It remains to specify the placement of each copy $L_t^s$ so as to guarantee the desired properties of $L_{t+1}$.

By induction, the links within each copy of $L_t$ are $f$-independent. We place the copies $L_t^s$ so that the links between any two copies are $f$-independent and are $f$-independent from $j_{t+1}$. Let $\ell_t^{s}=diam(L_t^s)=8^s\ell_t$ denote the diameter of $L_t^s$. 
Let $g(x)=C\log^{1/\alpha}x$. We place each copy $L_t^s$ at a distance $d(L_t^s,j_{t+1})=9\ell_t^sg(l_{j_{t+1}}/\ell_t^s)$ from $j_{t+1}$. The construction is ready.

We first prove the property (\ref{EN:independence}).
\begin{claim}
With the distances defined as above, the set $L_{t+1}$ is $f$-independent.
\end{claim}
\begin{proof} Since the links are arranged linearly, the maximum of $d_{ij}, d_{ji}$ for every pair of links $i,j$ is at least $\max\{l_i,l_j\}$. Hence, it suffices to prove that for any pair of links $i,j$ with $l_i\ge l_j$, $d(i,j)>9l_j f(l_i/l_j)$ (recalling that $\beta^{1/\alpha}=3$).
Consider any link $i\in L_t^s$. We have that 
\[
d(i,j_{t+1})\ge d(L_t^s,j_{t+1})=9\ell_t^sg(l_{j_{t+1}}/\ell_t^s)\ge 9l_ig(l_{j_{t+1}}/l_i)\ge 9l_if(l_{j_{t+1}}/l_i),
\]
where the second inequality follows from the fact that $xg(c/x)$ is an increasing function of $x$ and that $l_i<\ell_t^s$, and the third inequality follows because $f(x)\leq g(x)$ for all $x$. Thus, all the links in $L_t^s$ are $f$-independent from $j_{t+1}$. Now, let us show that any two links $i,k$ with $l_i\leq l_k$ from different copies $L_t^s$ and $L_t^r$ with $s > r$ are $f$-independent (no matter which link is from which copy). Since $f(x)\leq g(x)$, it will be enough to show that
\begin{equation}\label{E:pessdist}
d(i,k)> 9l_ig(l_k/l_i).
\end{equation}
 Recall that $xg(c/x)$ is an increasing function of $x$. Then, for a fixed $k$, the right-hand side of (\ref{E:pessdist}) is maximized when $l_i$ is maximum. On the other hand, for a fixed $i$,  the value $g(l_k/l_i)$ is maximized when $l_k$ is maximum, because $g$ is an increasing function. Let $j_t$ denote the maximum length link in $L_t$. Then, the maximum link length in $L_t^s$ (in $L_t^r$) is $8^sl_{j_{t}}$ ($8^rl_{j_{t}}$, resp.). Therefore, it is enough to show that 
\[
d(i,k)>9\ell_t^rg(8^sl_{j_t}/(8^rl_{j_t}))=9\ell_t^rg(8^{s-r})=9C(3(s-r))^{1/\alpha}\ell_t^r.
\]
We have that
\[
d(i,k) \geq d(L_t^s,L_t^r)=d(L_t^s,j_{t+1}) - d(L_t^r,j_{t+1}) - \ell_t^r\ge 9\ell_t^sg(l_{j_{t+1}}/\ell_t^s) - 10\ell_t^rg(l_{j_{t+1}}/\ell_t^r).
\]
The term $g(l_{j_{t+1}}/\ell_t^r)$ can be bounded as (using $\alpha\ge 1$)
\[
g(l_{j_{t+1}}/\ell_t^r)=g(8^{s-r}l_{j_{t+1}}/\ell_t^s) < 3(s-r)\cdot g(l_{j_{t+1}}/\ell_t^s), \mbox{hence}
\]
\begin{align*}
d(i,k) &\ge 9\ell_t^s g(l_{j_{t+1}}/\ell_t^s) - 30(s-r)\ell_t^rg(l_{j_{t+1}}/\ell_t^s)\\
& > C(9\cdot 8^{s-r} - 30(s-r))\ell_t^r \\
& > 27C(s-r)\cdot\ell_t^r.
\end{align*}
\end{proof}
Next, observe that (the first line follows because the links are arranged linearly) 
\begin{align}
\nonumber \ell_{t+1}&=l_{j_{t+1}}+d(L_t^k,j_{t+1}) + \ell_t^k\\
\nonumber &\leq l_{j_{t+1}}+9\ell_t^kg(l_{j_{t+1}}/\ell_t^k) + \ell_t^k\\
\nonumber &=8^{k+1}\ell_t + 9\cdot 8^kg(8)\ell_t + 8^k\ell_t\\
&=O(8^{ 2^{c\ell_t}}).
\end{align}
Since the minimum link-length in $L_{t+1}$ is $1$, we can conclude that $\Delta(L_{t})\le\ell_t\leq 2\uparrow (c_1t)$ for a constant $c_1$ and for each $t$, where $\uparrow$ denotes the tower function.
This implies that $t=\Omega(\log^*{\Delta(L_t)})$. The property (\ref{EN:lowerbound}) is now proven.

It remains to check (\ref{EN:inconsistency2}). Consider a link $i_s$ from $L_{t}^s$ where $i_s$ is the copy of a link $i\in L_{t}$. We have that
\[
d(i_s,j_t)\leq \ell_t^s+d(L_t^s,j_t)= \ell_t^s + 9C\ell_t^s\log^{1/\alpha}{(l_{j_{t+1}}/\ell_t^s)}\leq c_2\ell_t^s(k-s+1)^{1/\alpha},
\]
for a constant $c_2$. This implies:
\[
I(i_s,j_{t+1})=\left(\frac{l_{i_s}}{d(i_s,j_{t+1})}\right)^\alpha \geq \left(\frac{l_{i_s}}{  c_2(k-s+1)^{1/\alpha}\ell_{t}^s} \right)^\alpha \geq \frac{1}{c_3(k-s)\ell_{t-1}},
\]
where we used the fact that $l_{i_s}/\ell_t^s=l_i/\ell_t\geq 1/\ell_t$.
Now, let $\{i_s\in L_{t}^s | s=1,2,\dots,k\}$ be a set of links where $i_s\in L_{t}^s$, but they are not necessarily the copies of the same link of $L_{t}$. Then, 
\[
I(S,j_{t+1}) = \sum_{s=1}^{k}I(i_s,j_{t+1}) > \sum_{1}^{k}{\frac{1}{c_3(k-s+1)\ell_{t}}} = \Omega\left(\frac{\log{k}}{\ell_{t}}\right).
\]
Recall that $k=2^{c\ell_t}$. By choosing a large enough constant $c$, we can thus guarantee the property (\ref{EN:inconsistency2}). This completes the proof of all the properties of $L_t$ and the proof of the theorem.
\end{proof}

\subsection{Conflict Graphs without Power Control}
\label{ss:uniformlb}

Our results thus far show that conflict graphs can be used to obtain good approximation for scheduling problems that allow power control. That turns out not to be the case when power control is not available, that is, when we have the fixed uniform power assignment $P_0$: If there is no power control,  there is no conflict graph sandwich with tightness smaller than $\Theta(\log\Delta/\log\log\Delta)$. This claim is in contrast with the special case of unit length links (and uniform thresholds), where simple disk-graphs provide constant-tightness sandwiching~\cite{us:talg12}.

We prove the claim for linear (i.e., 1-dimensional) instances  with $\alpha=2$ and uniform thresholds $\beta=1$. It is not hard to show that uniform power scheduling is equivalent to its bidirectional variant (up to constant factors)~\cite{tonoyan11}, where we replace the distances $d_{ij},d_{ji}$ in the SINR formula with $d(i,j)$. Hence, we consider any conflict graph formulation $\cG$ that is, in view of the observations made in Sec.~\ref{ss:necessity}, of the following form: For every pair $i,j$ of links, they are \emph{independent in $\cG$} if $d(i,j)\ge c_1f(l_i,l_j)$, and are \emph{adjacent in $\cG$} if $d(i,j)<c_2f(l_i,l_j)$, where $c_1,c_2$ are constants and $f$ is an arbitrary function. The values of constants $c_1,c_2$ will not be important, so assume, for simplicity, that $c_1=c_2=1$. 

First, let  us show that there is a constant $h>0$, such that for every $\ell>0$, $f(\ell, \ell) \le h\ell$. It is an easy special case of Lemma~\ref{P:oblcore} that for some constant $h'>0$, every set of links of length $\ell$ arranged linearly with distance $d(i,j)=h'\ell$ between consecutive links is $P_0$-feasible. Hence, if $f(\ell, \ell) > h\ell$ for a number $h>h'$, then the tightness of the graph formulation $\cG$ is at least $\lfloor h/h'\rfloor$. On the other hand, we have $\Delta=1$ for the described instances, which means that the tightness must be bounded by a constant (in the context of our main claim), which implies that $h$ is bounded by a constant.

Next, we  bound from below $f(\ell_0,\ell_1)$, for any $\ell_0>\ell_1$. Consider a link $0$ of length $l_0=\ell_0$ and a large number of links $\{1,2,\dots,k\}$ of length $\ell_1< \ell_0$ arranged on the line such that $s_0$ is at the origin, $r_0$ is at coordinate $\ell_0$, and for $i=1,2,\dots,k$, $s_i$ is at $r_0+f(\ell_0,\ell_1)+(i-1)(h+1)\ell_1$, and $r_i$ is at $s_i+\ell_1$. Thus, $d(0,i)=f(\ell_0,\ell_1)+(i-1)(h+1)\ell_1$, and the spacing between any two links of length $\ell_1$ is at least $h\ell_1$, so the constructed set is independent in $\cG$, and must be feasible. The total interference-to-signal ratio on link $0$ is
\begin{align*}
I_0(\{1,2,\dots,k\}, 0) &=\sum_{i\ge 1} \frac{\ell_0^2}{(f(\ell_0,\ell_1)+(i-1)(h+1)\ell_1)^2} \\
&\ge \int_{0}^\infty \frac{\ell_0^2}{(f(\ell_0,\ell_1)+(h+1)\ell_1 x)^2}\ dx \\
&= \frac{\ell_0^2}{(h+1)\ell_1} \cdot \frac{1}{f(\ell_0,\ell_1)+(h+1)\ell_1}\ .
\end{align*}
For feasibility, the right-hand side must be less than 1, i.e., $f(\ell_0,\ell_1) > \frac{\ell_0^2}{(h+1)\ell_1}-(h+1)\ell_1$ must hold for every pair of distances $\ell_0,\ell_1$.

Now, we can use the obtained bound on $f(\ell_0,\ell_1)$ to construct an instance that is a clique in $\cG$, but is feasible with uniform power. For a number $n>4(h+1)^2$, consider the set ${1,2,\dots,k}$ of $k=n/\log n$ links, arranged on the line in the order $1,2,\dots,k$, such that $l_i=n^i$, and the minimum distance between links $i$ and $i+1$ is $$d(i,i+1)=\frac{l_{i+1}^2}{(h+1)l_i}-(h+1)l_i<n^{i+2}/(h+1).$$ Now, let us show that the obtained conflict graph has large chromatic number. Due to symmetry, it suffices to consider the conflicts with the longest link $k$. The link $k-1$ is adjacent to $k$, by the definition of the distances above. For each $i<k-1$, assuming $n$ is sufficiently large, we have 
$$d(i,k)<\sum_{t=1}^{k-1}\frac{n^{t+2}}{h+1} + n^{t}<\frac{2n^{k+1}}{h+1} + n^{k-1}<\frac{3n^{k+1}}{h+1}<\frac{l_{k}^2}{(h+1)l_i}-(h+1)l_i,$$
 where in the last two inequalities we used $n>4(h+1)^2$. This means that the obtained conflict graph is a clique of size $k$. On the other hand, it is easy to see that the instance is feasible: For every pair of links $i,j$ with $i>j$, the interference to signal ratio is at most $$\frac{l_i^2}{d(i,j)^2}\le \frac{l_i^2}{d(i,i-1)^2}\le \frac{n^{2i}}{(n^{i+1}/(h+1)-(h+1)n^{i-1})^2}\le \frac{1}{n},$$ since $n>4(h+1)^2$.

Finally, note that for the constructed instance, $\Delta=2^n$, which means that the approximation ratio provided by any conflict graph is at least $\Omega(\log\Delta/\log\log\Delta)$.

\subsection{General Metric Spaces}
\label{ss:genmetrics-lb}

The following proposition shows that conflict graphs can be arbitrarily poor approximation of the SINR model in general metric spaces. Given a function $f$, the construction consists of an $f$-independent set of \emph{unit length} links. Since all links have length $1$, $f$-independence is equivalent to $f(1)$-independence (that is, $g$-independence, where $g(x)\equiv f(1)$). The separation between the links is just enough to ensure $f(1)$-independence. However, since all the links are equally ($f(1)$-) separated from any given link, their interference accumulates and only a constant number of links can be scheduled in the same slot. This leads to schedules of length $\Theta(n)$. 
\begin{proposition}
 For every positive function $f$ and any $n\geq 1$, there is an $f$-independent set of $n$ unit length links (i.e., $\Delta=1$) that cannot be partitioned into less than $\Theta(n)$ feasible subsets, under uniform thresholds $\beta\ge 1$.
\end{proposition}

\begin{proof}
Let $L=\{1,2,\dots,n\}$ be a set of links of unit length. 
The distance between every two senders of links is the same: $d(s_i, s_j) = 2\beta f(1)$.
Distances to and between receivers are then induced by these distances and lengths; e.g., distances between receivers is $d(r_i, r_j) = d(s_i,s_j) + l_i + l_j = 2\beta f(1)+2$.
The set $L$ is $f$-independent, since $d(i,j)>\beta f(1)\cdot l_i=\s_if(\s_j/\s_i)$. 
Consider any $P$-feasible subset $S$ of $k$ links for a power assignment $P$, and fix a link $i\in S$. The SINR condition implies that $P(i) > \beta \sum_{j\in S\setminus\{i\}}\frac{P(j)l_{i}^\alpha}{d_{ji}^\alpha}$ and $P(j) > \beta \frac{P(i)l_j^\alpha}{d_{ij}^\alpha}$ for all $j\in S\setminus \{i\}$.  Substituting for $P(j)$ in the first inequality and canceling the term $P(i)$, we obtain:
\[
1 >\sum_{j\in S\setminus\{i\}} \beta^2 {\frac{l_i^{\alpha }l_j^\alpha}{d_{ij}^\alpha d_{ji}^\alpha}}=\beta^2 \sum_{j\in S\setminus\{i\}}{\frac{1}{(2f(1)+1)^2}}=\frac{(|S|-1)\beta^2}{(2f(1)+1)^2},
\]
which implies that $|S| < \left(2f(1)+1\right)^2/\beta^2+1=O(1)$. Since $S$ was an arbitrary feasible subset of $L$, we conclude that $L$ cannot be split into fewer than $\Theta(n)$ feasible subsets.
\end{proof}

\subsection{Noise-Limited Networks}
\label{ss:weaklb}

\newcommand{\pmax}{P_{max}}
\newcommand{\lmax}{l_{max}}
\newcommand{\lhat}{\hat{l}}
\newcommand{\lmin}{l_{min}}

Recall that in order to obtain our approximations, we assumed in Sec~\ref{s:model} that there is a constant $c>1$, such that for each link $i$, $P(i)\ge c N_i \s_i^\alpha$. However, this is not always achievable when nodes have limited power. Suppose that each sender node has maximum power $\pmax$. For concreteness, we assume that $c=2$, $N_i=N>0$, $\beta_i=1$, for all links $i$, and the links are in a Euclidean space. Thus, a link $i$ is \emph{weak} if $\pmax \le 2 N l_i^\alpha$. Note that a link is weak because it is too long for its maximum power, i.e.\ $l_i \ge \lmax/2^{1/\alpha}$, where  $\lmax=(\pmax/ N)^{1/\alpha}$ is the maximum length a link can have to be able to overcome the noise when using maximum power. Scheduling weak links may be considered as a separate problem. Let $\tau$-{\wscheduling} denote the problem of scheduling weak links using power assignment $P_{\tau}$. We show here that the problem of scheduling (not necessarily weak links) with uniform power assignment (i.e., $P_0$), denoted {\uscheduling}, can be reduced to $\tau$-{\wscheduling} for any given $\tau\in [0,1]$, modulo constant approximation factors. Namely, a $\mu$-approximation algorithm for $\tau$-{\wscheduling} can be turned into a $O(\mu)$-approximation algorithm for {\uscheduling}. To our knowledge, there is no known approximation algorithm for {\uscheduling} with ratio in $o(\min(\log\Delta, \log n))$.

\begin{theorem}\label{T:weakhard}
There is a polynomial-time reduction from {\uscheduling}
to $\tau$-{\wscheduling} for any $\tau\in [0,1]$, that preserves approximation ratios up to constant factors.
\end{theorem}

The proof directly follows from the two Lemmas below.
\begin{lemma}
There is a polynomial-time reduction from $0$-{\wscheduling} to $\tau$-{\wscheduling}, with any given $\tau\in [0,1]$, that preserves approximation ratios up to constant factors.
\end{lemma}
\begin{proof}
Consider a $P_{\tau}$-feasible set $S$ of weak links. It is enough to show that $S$ can be partitioned into a constant number of $\pmax$-feasible subsets. Recall that $S$ is $P_{\tau}$-feasible  if for each link $i\in S$, $\frac{P_{\tau}(i)}{l_i^\alpha} > \sum_{j\in S\setminus i} \frac{P_\tau(j)}{d_{ji}^\alpha} + N$, or equivalently,  
\[
\frac{\pmax}{l_i^\alpha} > \sum_{j\in S\setminus i} \frac{P_\tau(j)}{P_\tau(i)}\cdot \frac{\pmax}{d_{ji}^\alpha} + \frac{\pmax N}{P_\tau(i)}.
\]
 Since the links are weak, we have $l_j/l_i\le 2^{1/\alpha}$, implying that $P_{\tau}(j)/P_{\tau}(i)\le 2^\tau$, and have $\pmax\le 2Nl_i^\alpha$, implying $\frac{\pmax}{P_{\tau}(i)}\le \frac{2Nl_{i}^\alpha}{Nl_i^{\alpha}}=2$, where we also used the fact that $P_{\tau}(i)\ge Nl_i^\alpha$, as $S$ is $P_{\tau}$-feasible. Hence, a $2$-strong subset of $S$, w.r.t. $P_\tau$, is $\pmax$-feasible. The proof is completed by recalling (Thm.~\ref{T:signalstrengthening}) that each $P_{\tau}$-feasible set can be partitioned into four $2$-strong subsets w.r.t. $P_\tau$.
\end{proof}

\begin{lemma}
There is a polynomial-time reduction from {\uscheduling} to $0$-{\wscheduling} that preserves approximation ratios up to constant factors.
\end{lemma}
\begin{proof}
We show that a given $\pmax$-feasible set $S$ of non-weak links can be transformed into a set $S'$ of weak links that can be partitioned into $O(1)$ subsets, each $\pmax$-feasible. Recall that set $S$ is $\pmax$-feasible if and only if $\sum_{j\in S\setminus i}\left(\frac{g_i}{d_{ji}}\right)^\alpha < 1$ holds for every link $i\in S$, where $g_i=\frac{l_i}{(1-\beta N l_i^\alpha/\pmax)^{1/\alpha}}=\frac{l_i}{(1-(l_i/\lmax)^\alpha)^{1/\alpha}}$. The idea is to apply a geometric transformation on the set $S$, so that every link becomes weak, while the ratios $\frac{g_i}{d_{ji}}$ change by no more than constant factors. To this end, we first scale the set of sender nodes in $S$ (taken as points in the space) by a factor $X>0$, then ``stretch'' each link separately, by moving only its receiver node.

Let $\lmin$ denote the smallest link length in $S$, and $\lhat=\lmax/2^{1/\alpha}$ denote the border link length between weak and non-weak links. 
We want to map the links with length in range $[\lmin, \lmax)$ to the range
$[\lhat, \lmax)$, as described above.
Denote $g(x)=\frac{x}{(1-(x/\lmax)^\alpha)^{1/\alpha}}$ the function that ``generates'' the coefficients $g_i=g(l_i)$. Since $g(x):(0,\lmax)\rightarrow (0,\infty)$ is a continuous and monotonically increasing function, so is its inverse $f=g^{-1}:(0,\infty)\rightarrow (0,\lmax)$. Now, the set $S'$ of links is constructed as follows. To each link $i\in S$  corresponds a single link $i'\in S'$. The sender node $s_{i'}$ is located at the point $r_{i'}=X\cdot s_i$ with $X=\lmax/\lmin$. The receiver node $r_{i'}$ is located at $r_{i'}=s_{i'}+f(Xl_i)\cdot(r_i-s_i)/l_i$. Thus, $l_{i'}=f(Xl_i)<\lmax$. Also, the facts that $g(\lhat)=2^{1/\alpha}\lhat =\lmax\le Xl_i$ and  that $f$ is an increasing function, imply that $l_{i'} = f(X l_i) \ge f(g(\lhat)) = \lhat$, that is, $i'$ is indeed a weak link.

In order to complete the proof, we need to show that $S'$ can be split into a constant number of feasible subsets. To this end, we first use Thm.~\ref{T:signalstrengthening} to split $S$ into at most $\lceil 2\cdot 4^\alpha\rceil$ subsets, each $4^\alpha$-strong. Let $T$ be one of those. It suffices to show that $T'\subseteq S'$, the image of $T$ under our mapping, is feasible. Let $i,j\in T$ be any pair of links. 

First, note that since $i$ is a non-weak link, $g_i\in [l_i,2^{1/\alpha}l_i]$, and by the choice of the length transformations, $g_{i'}=g(f(Xl_i))=X l_i\le 2^{1/\alpha}X g_i$. Next, we show that $d_{j'i'}\ge Xd_{ji}/2$. Since $T$ is $4^\alpha$-strong, it is easy to show that $d_{ji}>4l_i$, which implies that $d(s_i,s_j)\ge d_{ji}-l_i> 3d_{ji}/4>3l_i$. By construction, $d(s_{i'},s_{j'})=X\cdot d(s_i,s_j)>3Xl_i\ge 3f(Xl_i)=3l_{i'}$, where we also used $f(x)\le x$ for all $x\in (0,\lmax)$, which follows from the fact that $g(x)\ge x$. Again, by the triangle inequality, 
\[
d_{j'i'}\ge d(s_{i'},s_{j'})-l_{i'}>2d(s_{i'},s_{j'})/3=2Xd(s_{i},s_{j})/3>Xd_{ji}/2.
\]
 Putting all together, we see that  $(g_{i'}/d_{j'i'})^\alpha \le (2\cdot 2^{1/\alpha}\cdot g_{i}/d_{ji})^\alpha\le 4^\alpha \cdot (g_{i}/d_{ji})^\alpha$. Since $T$ is a $4^\alpha$-strong set, this easily implies that $T'$ is feasible.
\end{proof}

\section{Context}
\label{s:context}

\subsection{Related Work}
 Gupta and Kumar introduced the SINR model of interference/communication in their influential paper~\cite{kumar00}.
Moscibroda and Wattenhofer \cite{moscibrodaconnectivity} initiated worst-case analysis of scheduling problems in networks of arbitrary topology, which is also the setting of interest in this paper.
There is a huge literature on wireless scheduling problems, but we focus here on algorithms with performance guarantees.

There has been significant progress during the past decade in understanding scheduling problems with fixed uniform data rates.
NP-completeness results have been given for different variants \cite{goussevskaiacomplexity, katz2010energy,lin2012complexity}.
Early work on approximation algorithms involve (directly or indirectly) partitioning links into length groups,
which results in performance guarantees that are at least logarithmic in
$\Delta$, the link length diversity: TDMA scheduling and uniform weights {\wcapacity} \cite{goussevskaiacomplexity,dinitz,us:talg12},
non-preemptive scheduling \cite{fu2009power}, joint power control, scheduling and routing \cite{chafekarcrosslayer},
and joint power control, routing and throughput scheduling in multiple channels \cite{AG10}, to name a few.
Constant-factor approximations are known for uniform weight {\wcapacity}, with uniform power \cite{GoussevskaiaHW14}, oblivious power \cite{SODA11}, and (general) power control \cite{kesselheimconstantfactor}.
The characterization of feasibility under general power in \cite{kesselheimconstantfactor} is essential for our results.
Standard approaches translate the constant-factor approximations for the uniform weight {\wcapacity}  into $O(\log n)$-approximations for TDMA link scheduling  and general {\wcapacity}.
On the other hand, \cite{wan09,jansen03} present approximation-preserving (up to constant factors) reductions from the fractional scheduling and routing and scheduling problems to {\wcapacity}, which in combination with the results above gives us $O(\log n)$-approximations for those problems (note, however, that this reduction uses computationally heavy linear programming techniques). 
The observation that extending  inductive independent graphs to the multi-channel multi-radio case essentially preserves inductive independence (Sec.~\ref{s:problems}) has been made in~\cite{WanCWY11}. The $O(\log n)$-approximation results do not require the assumption we make regarding interference-constrained networks, and some also work in general metrics.

Algorithms for the graph-based variants of flow routing and scheduling problems have initially been addressed in \cite{KodialamN03,Kumar2005,AlicherryBL06,WanCWY11,wan14}, among others. Algorithms (based on the primal-dual method) with performance guarantees in terms of inductive independence are presented in \cite{wan14}. Algorithms with performance guarantees for the graph-based variant of combinatorial auctions are presented in~\cite{HoeferKV14,HoeferK15}. Those works also present algorithms for the SINR model, with an extra $O(\log n)$  approximation factor.
Many problems become easier in the regime of linear power assignments, and constant factor approximations are known for 
{\wcapacity} \cite{wang2011constant,halmitcognitive} and 
TDMA link scheduling \cite{fangkeslinear,tonoyanlinear}.

The communication ability of packet networks is characterized by the capacity region, i.e.\ the set of data rates that can be supported by any scheduling policy. In order to achieve \emph{low delays} (i.e.,\ polynomially bounded queues) and \emph{optimal throughput}, the classic result of Tassiulas and Ephremides \cite{TE92} and followup work in the area (e.g.~\cite{LinShroff04}) imply that {\mwisl} is a core optimization problem that lies at the heart of such questions. This reduction applies to very general settings involving  single-hop and multi-hop, as well as fixed and controlled transmission rate networks. Moreover, approximating {\wcapacity} within any factor implies achieving the corresponding fraction of the capacity region. 
In general, even approximating the capacity region  in polynomial time within  a non-trivial bound, while keeping the delays low, is hard under standard complexity-theoretic assumptions~\cite{ShahTT11}.
Methods with good performance, such as Carrier Sense Multiple Access (CSMA) \cite{jiang2010distributed} necessarily require exponential time in the worst case \cite{JiangLNSW12}.

Very few results on performance guarantees are known for problems involving rate control.
The constant-factor approximation for  {\wcapacity} with uniform weights and arbitrary but fixed data rates proposed in \cite{KesselheimESA12} can be used to obtain  $O(\log n)$-approximations for TDMA link scheduling and {\wcapacity} with rate control, where $n$ is the number of links.
A recent work~\cite{goussevskaia2016wireless} handles the TDMA scheduling problem with fixed but different rates, obtaining an approximation independent of the number of links $n$, but the ratio is polynomial in $\Ds$.


The idea of modeling SINR with graphs arose early. Disc graphs were shown to be insufficient in general \cite{moscibrodaconnectivity}. 
However, it was observed rather early that equilength sets of links can be captured with unit disc graphs \cite{goussevskaiacomplexity}. In fact, a sandwiching result with constant tightness holds for equilength links \cite{us:talg12}. For links of widely varying lengths, less was known: a $O(\log\log \Delta \log n)$-tight construction was given in \cite{us:talg12}.

\subsection{Modeling Issues}

The SINR model has been an object of intense study given its closer fidelity to reality than binary models.
Indeed, the additivity of interference and the near-threshold nature of signal
reception has been well established in experiments.  
The model is though far from perfect: the assumption of signal strength decreasing inversely polynomially with distance can be far off \cite{son2006,MaheshwariJD08,sevani2012sir,us:MSWiM14}.  
We discuss here the various proposed alternatives and explain why analysis in the pure SINR model is of fundamental importance.

In stochastic analysis (see e.g.\ \cite{haenggi2009}), as well as in
simulation studies, the canonical approach is to assume stochastic
\emph{fading} or \emph{shadowing}, where signal strengths include a
multiplicative random component.  Such stochastic components are a
natural fit for stochastic analysis, but less so for the every-case analysis
aimed for here. The stochastic models \emph{seem} to generate
instances that are similar to real ones, but little rigorous validation
exists.  The question is then what we can say about the instance at
hand, rather than some distribution.

One can distinguish between two types of stochastic effects: time-varying \emph{fading} and time-invariant \emph{shadowing}.
It is typically assumed that shadowing is independent across space, and fading is independent across time.
The only work we are aware of involving performance guarantee analysis in the presence of shadowing is our recent work \cite{us:mobihoc17}.
It suggests that the usual assumptions about independence of the random variables across space has a major effect, as it can lead to counter-intuitive improvements in the size of optimal solutions. Much more remains though to be considered on this front.

For time-varying effects under Rayleigh fading, it has been shown that
applying algorithms based on the deterministic formula results in
nearly equally good results \cite{dams2015}.  In fact, for {\mwisl},
this only affects the constant factor \cite{us:mobihoc17}.  Thus,
asymptotic results in the standard non-fading model carry fully over
to settings with Rayleigh fading, including our approximation ratios.

For every-case analysis, a natural generalization of the SINR model
would be to shed the geometry and allow for an arbitrary signal-quality matrix.
One could in practice obtain this in the form of facts-on-the-ground signal
strength measurements \cite{us:MSWiM14,us:PODC14}.  
This generalization is, however, too expensive as it runs into the
computational intractability monster: with such a formulation one can
encode the coloring problem in general graphs \cite{GoussevskaiaHW14}, which is known to be famously hard to approximate \cite{FeigeKilian}.

A more moderate approach is to relax the Euclidean assumption to more
general metric spaces, as first proposed in \cite{fangkeslinear}.
We assume here doubling metrics \cite{us:talg12}, which has been a standard assumption 
when dealing with problems beyond unweighted throughput.

Alternatively one can analyze algorithms in terms of some parameters of the signal-quality matrix. Such results then apply directly to the SINR model, but do not depend on the exact features of the model.
The most successful such effort has involved the so-called \emph{inductive independence} number, proposed in \cite{HoeferK15}, which has been applied for spectrum auctions \cite{HoeferK15,HoeferKV14}, dynamic packet scheduling \cite{kesselheimStability}, online independent sets \cite{GHKSV14}, and connectivity \cite{HHMW17}. 
The parameter is known to be a constant in SINR settings with power control \cite{HHMW17}.

  Another parameter is \emph{$C$-independence}, proposed in 
\cite{dams2014jamming}, based on a formulation in \cite{infocom11}. 
It is constant-bounded in SINR models with uniform power \cite{infocom11}.
It has been applied to the distributed optimization of (uniform weights) {\mwisl} via learning \cite{infocom11}, and its extensions involving jamming \cite{dams2014jamming} and channel availabilities \cite{dams2013sleeping}.
Both parameters, inductive independence and $C$-independence, however, are useful for unweighted throughput maximization, but have failed to give sublogarithmic bounds for weighted throughput or scheduling latency minimization thus far.

Ultimately, the pure SINR model lies at the core of all these models. It is exact in free space, forms the base case under stochastic fading, and is the starting point for any of the worst-case models. It is essential to understand properly how this fundamental case works, and then relax the assumptions as much as possible. It does appear that the doubling metrics we use are necessary for the results of the kind that we obtain. It remains to be seen what can be done in other settings.

\section{Conclusions}
Our results suggest that a reassessment of the role of graphs as wireless models might be in order.
By paying a small factor (recalling, as well, that $\log^*(x) \le 5$ in this universe),
we can work at higher levels of abstraction, with all the algorithmic and analytical benefits that it accrues.
At the same time, hopes for fully constant-factor approximation algorithms for core scheduling problems may have somewhat receded.

It would be interesting to see if other natural classes of hypergraphs admit efficient sketches. It would also be interesting to explore further properties of generalized disk graphs.

\section*{Acknowledgements}

We are grateful to Eyj\'olfur Ingi \'Asgeirsson for collaborations and experimentation.
We thank Allan Borodin, Guy Even, Stephan Holzer, Calvin Newport and Roger Wattenhofer for helpful discussions.

\bibliographystyle{abbrv}
\bibliography{Bibliography}

\end{document}